\documentclass[12pt]{article}
\usepackage[authoryear,round]{natbib}
\usepackage{amssymb}
\usepackage{graphicx}
\usepackage{latexsym}
\usepackage[left=1in,top=1in,right=1in,bottom=1in]{geometry}

\usepackage{setspace}

\usepackage{xcolor}
\definecolor{Bleu}{RGB}{0,0,204}
\definecolor{Violet}{RGB}{102,0,204}
\definecolor{Rouge}{RGB}{204,0,0}
\definecolor{Highlight}{RGB}{251,0,0}
\usepackage[hypertexnames=false]{hyperref}
\hypersetup{
colorlinks,
  citecolor=Bleu,
  linkcolor=Bleu,
  urlcolor=Violet} 

\usepackage{amsmath}
\usepackage{amsthm,thmtools}
\renewcommand\thmcontinues[1]{Continued}
\usepackage{amssymb} 
\usepackage{booktabs}
\usepackage{color}
\usepackage{tikz}
\usepackage{rotating}

\usepackage{enumitem}
\makeatletter
\newcommand{\mylabel}[2]{
	\addtocounter{\@listctr}{-1}%
    \refstepcounter{\@listctr}%
	#2\def\@currentlabel{#2}\label{#1}}
\makeatother

\usepackage{algorithm}
\usepackage{algorithmicx}
\usepackage[noend]{algpseudocode}

\usepackage{mathrsfs} 

\usepackage{dsfont} 

\usepackage{accents}

\usetikzlibrary{shapes}
\usetikzlibrary{arrows}
\definecolor{darkblue}{rgb}{0,0.4,0.9}
\definecolor{gray10}{rgb}{0.1,0.1,0.1}
\definecolor{gray20}{rgb}{0.2,0.2,0.2}
\definecolor{gray30}{rgb}{0.3,0.3,0.3}
\definecolor{gray40}{rgb}{0.4,0.4,0.4}
\definecolor{gray60}{rgb}{0.6,0.6,0.6}
\definecolor{gray80}{rgb}{0.8,0.8,0.8}
\definecolor{gray90}{rgb}{0.9,0.9,.9}
\definecolor{gray95}{rgb}{0.95,0.95,.95}
\definecolor{gray96}{rgb}{0.96,0.96,.96}
\definecolor{lgreen} {RGB}{180,210,100}
\definecolor{dblue}  {RGB}{20,66,129}
\definecolor{ddblue} {RGB}{11,36,69}
\definecolor{lred}   {RGB}{220,0,0}
\definecolor{nred}   {RGB}{224,0,0}
\definecolor{norange}{RGB}{230,120,20}
\definecolor{nyellow}{RGB}{255,221,0}
\definecolor{ngreen} {RGB}{98,158,31}
\definecolor{dgreen} {RGB}{78,138,21}
\definecolor{nblue}  {RGB}{28,130,185}
\definecolor{jblue}  {RGB}{20,50,100}
\definecolor{nnyellow}{RGB}{235,200,0}
\definecolor{purple}{RGB}{150, 0, 120}
\definecolor{sgGreen} {RGB}{20, 180, 50}
\definecolor{revised}{rgb}{0,0,0.9}

\newtheorem{theorem}{Theorem}

\newtheorem{lemma}[theorem]{Lemma}

\theoremstyle{definition}

\newcommand{\openr}{\hbox{${\rm I\kern-.2em R}$}}
\newcommand{\openn}{\hbox{${\rm I\kern-.2em N}$}}
\newcommand{\logit}{\operatorname{logit}}
\newcommand{\expit}{\operatorname{expit}}
\newcommand{\Rem}{\operatorname{Rem}}
\newcommand\independent{\protect\mathpalette{\protect\independenT}{\perp}}
\def\independenT#1#2{\mathrel{\rlap{$#1#2$}\mkern2mu{#1#2}}}
\newcommand{\norm}[1]{\left\lVert#1\right\rVert}

\newcommand{\Id}{\textnormal{Id}}

\newcommand{\mb}[1]{\mathbf{#1}}

\DeclareMathOperator{\Var}{Var}

\DeclareMathOperator{\E}{\mathbb{E}}

\DeclareMathOperator{\Ind}{\mathds{1}}

\DeclareMathOperator{\RR}{RR}

\newcommand{\indicator}[1]{\mathds{1}\{ #1 \}}
\makeatletter
\newlength{\trianglerightwidth}
\algnewcommand{\LineCommentCont}[1]{\Statex \hskip\ALG@thistlm%
  \parbox[t]{\dimexpr\linewidth-\ALG@thistlm}{\hangindent=\trianglerightwidth \hangafter=1 \strut$\triangleright$ #1\strut}}
\makeatother

\allowdisplaybreaks

\usepackage[normalem]{ulem}
\usepackage{longtable}
\usepackage{authblk}

\title{Efficient Principally Stratified Treatment Effect Estimation in Crossover Studies with Absorbent Binary Endpoints}
\author[1,2]{Alex  Luedtke}

\affil[1]{\footnotesize   Department of Statistics, University of Washington, USA}
\affil[2]{\footnotesize   Vaccine  and   Infectious  Disease   Division,  Fred
  Hutchinson Cancer Research Center, USA} 

\author[3]{Jiacheng Wu}
\affil[3]{\footnotesize   Department of Biostatistics, University of Washington, USA} 

\date{}

\begin{document}
\maketitle
\begin{abstract}\singlespacing
Suppose one wishes to estimate the effect of a binary treatment on a binary endpoint conditional on a (possibly continuous) post-randomization quantity in a counterfactual world where all subjects received treatment. It is generally difficult to identify this parameter without strong, untestable assumptions. It has been shown that identifiability assumptions become much weaker under a crossover design where (some subset of) subjects not receiving treatment are later given treatment. Under the assumption that the post-treatment biomarker observed in these crossover subjects is the same as would have been observed had they received treatment at the start of the study, one can identify the treatment effect with only mild additional assumptions. This remains true if the endpoint is absorbent, i.e. an endpoint such as death or HIV infection such that the post-crossover treatment biomarker is not meaningful if the endpoint has already occurred. In this work, we review identifiability results for a parameter of the distribution of the data observed under a crossover design with the principally stratified treatment effect of interest. We describe situations where these assumptions would be falsifiable given a sufficiently large sample from the observed data distribution, and show that these assumptions are not otherwise falsifiable. We then provide a targeted minimum loss-based estimator for the setting that makes no assumptions on the distribution that generated the data. When the semiparametric efficiency bound is well defined, for which the primary condition is that the biomarker is discrete-valued, this estimator is efficient among all regular and asymptotically linear estimators and is more efficient than existing nonparametric approaches in situations where a continuous baseline covariate is predictive of either the outcome or the biomarker. We also present a version of this estimator for situations where the biomarker is continuous. Implications to closeout designs for vaccine trials are discussed.
\end{abstract} 

\noindent%

\section{Introduction}\label{sec:intro}
Suppose one wishes to assess the effect of a binary treatment on a binary endpoint. For simplicity, we refer to the levels of the treatment as ``treated'' and ``untreated'', but one could alternatively study a contrast between two active treatments. The objective is to develop a biomarker that explains much of the treatment effect variation, i.e. a variable for which the stratified treatment effect is highly variable. Sometimes, the most predictive biomarkers are defined using the counterfactual values that a post-treatment biomarker would take had a subject been treated or untreated. Conditioning on these biomarkers was termed principal stratification by \cite{Frangakis&Rubin2002}. Principally stratified analyses are especially interesting in situations where one wishes to bridge a treatment's efficacy from one population to another, and so treats a small subset of the new population and observes their biomarkers in order to predict what the treatment effect would be in this population \citep{Gilbertetal2011}. Though additional transportability assumptions \citep{Bareinboim&Pearl2012} are needed to be able to bridge a treatment effect, these assumptions become more plausible if a strong effect modifier is available. In this paper, we focus only on counterfactual post-randomization biomarkers that occur under treatment. This ensures that one could observe the counterfactual biomarker under an appropriate intervention -- this is in contrast to stratifying on the biomarker's counterfactual value both after treatment and after a lack of treatment. Furthermore, in vaccine studies, which represent an important application area, the post-treatment biomarkers of most interest are immune responses, and in many settings untreated subjects who have not experienced the endpoint will have no immune response, so that studying treatment effects conditional on unvaccinated immune responses is uninteresting. We assume that the endpoint is absorbent \citep{Nason&Follmann2010}, so that the biomarker has no substantive meaning once the endpoint has occurred. For example, if the treatment is vaccination, the biomarker is vaccine-induced immune response, and the endpoint is HIV infection, then a measured immune response is only of interest if it precedes HIV infection. For most biomarkers of interest, death is also an absorbent endpoint.

Under most study designs, identifying principally stratified treatment effects requires strong, untestable assumptions. In the context of vaccine trials, \cite{Follmann2006} demonstrated the utility of crossover designs \citep[e.g.,][]{Woodsetal1989} for dramatically weakening these needed assumptions, where here the designs were referred to as closeout placebo vaccination designs and featured a crossover of uninfected placebos to the vaccine. There have since been numerous works that study closeout designs \citep{Gilbert&Hudgens2008,Wolfson&Gilbert2010,Gabriel&Gilbert2013,Huangetal2013}. 
Others have demonstrated that crossover designs are also of interest for estimating principally stratified treatment effects for treatments other than vaccines \citep{Wolfson&Henn2014,Gabriel&Follmann2016}.  All of these earlier works either relied on baseline covariates being discrete or on correctly specified (semi)parametric models. In this work, we provide an estimator of principally stratified treatment effects that is efficient within the nonparametric model that at most makes assumptions on the probability of receiving treatment given covariates. The efficiency of this estimator relies on the biomarker being discrete, since otherwise the semiparametric efficiency bound \citep{Pfanzagl1990,Bickel1993} will not be defined. We then generalize this nonparametric estimator to the case that the biomarker is continuous.

From a practical standpoint, closeout placebo vaccinations are easy to perform in a vaccine clinical trial setting because uninfected placebo recipients are still under follow-up at the end of the trial and are thus able to be vaccinated. Nonetheless, a downside to performing a vaccinating uninfected placebo recipients at the end of the study is that they are no longer available for additional follow-up as placebo recipients, making it difficult to assess the long-term efficacy of the vaccine. One option is to only vaccinate a (random) subset of uninfected placebo recipients at the end of a trial so that some placebo recipients are still available for longer-term follow-up.

We note that studying principally stratified treatment effects answers fundamentally different questions than does studying direct effects \citep{RobinsGreenland1992,Pearl2001,Vanderweele2015}, which would aim to study the effect of treatment on the outcome if the post-treatment biomarker were set to some fixed value or to the value it would have taken had the subject been untreated. This is in contrast to principal stratification analyses, which are useful for bridging vaccine efficacy from one population to another, provided one collects the post-treatment biomarker on a small subset of subjects in the new population and certain transportability assumptions hold. Identifying counterfactual direct effect estimands with the observed data distribution requires strong assumptions, including that there are no unmeasured confounders of the effect of the biomarker on the outcome \citep{Cole&Hernan2002}. In crossover studies, this no-unmeasured-confounders assumption does not need to hold to obtain valid inference for the principally stratified parameters that we study in this work. Nonetheless, we emphasize that estimation should be focused on the quantity that best reflects the scientific question. Furthermore, we note that, though principally stratified analyses focus on a different quantity than direct effect analyses, it has been shown that the presence of a nonzero principally stratified effect, conditional equality of the counterfactual biomarker under treatment and the counterfactual biomarker under a lack of treatment, implies the presence of a nonzero direct effect \citep{VanderWeele2008}; the reverse implication does not hold. In this work, we study effects stratified only on the counterfactual biomarker under treatment, so the result of \cite{VanderWeele2008} does not generally apply. Nonetheless, in settings in which the counterfactual biomarker under a lack of treatment is degenerate at some value $s_0$, e.g. in HIV vaccine studies where the biomarker is an immune response and there is no immune response if the vaccine is not administered, the presence of a nonzero treatment effect conditional on the counterfactual biomarker value under treatment equal to $s_0$ indicates the presence of a direct effect, i.e. an effect in a counterfactual world where treatment was administered and the biomarker had been set to $s_0$.

\noindent\textbf{Organization of Manuscript.} Section~\ref{sec:prelim} introduces the notation, parameters of the counterfactual distribution, identifiability of these parameters with parameters of the observed data distribution, and the statistical estimation problem. Section~\ref{sec:estDisc} describes our proposed estimator when the biomarker is discrete. Section~\ref{sec:2p} presents an extension of this estimation scheme to two-phase sampling studies. Section~\ref{sec:cont} presents an estimator that can be used when the biomarker is continuous. For simplicity, this estimator is presented in a single-phase sampling setup, but the extension to two-phase sampling is straightforward. Section~\ref{sec:sim} presents a simulation study. Section~\ref{sec:disc} closes with a discussion.

Appendix~\ref{app:proofs} provides proofs of the results from the main text. Appendix~\ref{app:cont} presents theorems and proofs for the validity of the estimator for continuous biomarkers presented in Section~\ref{sec:cont}. Appendix~\ref{app:EIF} gives a brief review of the derivation of efficient influence functions, which are used to define our estimators. Appendix~\ref{app:twophase} presents a TMLE that can be used when the data was derived from a two-phase sampling design. Appendix ~\ref{app:sim} gives additional simulation results.

\section{Notation and Identifiability}\label{sec:prelim}
Consider the counterfactual data structure $(W,A,S_1,Y_0,Y_1,S_0^c)$ drawn from a distribution $P^F$, where $W$ is a baseline covariate, $A$ is a treatment indicator, $S_1$ is a post-treatment biomarker, $Y_a$ is a counterfactual outcome of interest under treatment $a$ that takes on values zero and one, and $S_0^c$ is the counterfactual post-crossover biomarker had the subject not been treated at baseline and subsequently been crossed over to treatment. In a closeout placebo vaccination study, $S_1$ represents a post-vaccination immune response, and $S_0^c$ denotes the closeout biomarker measurements for placebos with $Y=0$, and $S_0^c=0$ for all other subjects. To simplify presentation, we assume that $S_1$ and $S_0^c$ are measured immediately following treatment. One could alternatively assume that no events occur before the time at which $S_1$ is measured \citep{Follmann2006}, or could assume some form of equal clinical risk, which roughly states that the risk of the event occurring before this time point is equal between treated and untreated subjects \citep{Wolfson&Gilbert2010}.

The objective is to estimate a contrast between the principally stratified mean outcome on treatment and in the absence of treatment, i.e. between $\E^F[Y_1|S_1=s_1^\star]$ and $\E^F[Y_0|S_1=s_1^\star]$, where here and throughout we let $\E^F$ denote an expectation under $P^F$ and we define $s_1^\star$ as a particularly interesting value of the biomarker $S_1$. For example, we may be interested in estimating the principally stratified relative risk 
\[
\RR^F(s_1^\star)\equiv \frac{\E^F[Y_1|S_1=s_1^\star]}{\E^F[Y_0|S_1=s_1^\star]}
\] 
or vaccine efficacy
\begin{equation}
\text{VE}^F(s_1^*) \equiv 1 - \frac{\E^F[Y_1|S_1=s_1^\star]}{\E^F[Y_0|S_1=s_1^\star]}
\label{equ:ve}
\end{equation}
One may also be interested in $\RR^F$ or $\text{VE}^F$ as a curve, across all values of $s_1^\star$. The developments in this work can also be used to estimate an additive contrast, thereby giving the conditional additive treatment effect.

In practice, we observe $(W,A,S,Y,S^c)\sim P$. We wish to identify the principally stratified mean outcomes with a parameter of this observed data distribution $P$. We make the following identifiability assumptions, where we note that here and throughout we use $\E$ to denote an expectation under $P$:

\begin{enumerate}[series=identassumptions,label=(A\arabic*),ref=A\arabic*]
	\item\label{it:sutvacons} A draw of $O=(W,A,S,Y,S^c)$ from $P$ has the same distribution as a draw of \allowbreak$(W,A,\Ind_{A=1}S_1,Y_A,\Ind_{A=0} S_0^c)$ from $P^F$.
\end{enumerate}
The above is implied by the consistency assumption, which states that $S_1=S$ if $A=1$, $Y=Y_A$, and $S^c=S_0^c$ if $A=0$. 
We also make the following assumptions:
\begin{enumerate}[resume=identassumptions,label=(A\arabic*),ref=A\arabic*]
	\item\label{it:ignorable} Ignorable treatment assignment: $A\independent (S_1,Y_0,Y_1) | W$, and
	\item\label{it:crossover} Crossover assumption: for $P^F$ almost all $w$, $S_1|Y_0=0,W=w$ under $P^F$ has the same distribution as $S_0^c|Y_0=0,W=w$ under $P^F$.
\end{enumerate}
This assumption is ill-defined when $P^F(Y_0=0|W=w)=0$, but we note that it could be replaced by equality of the conditional subdistribution functions $S_1,Y_0=0|W=w$ and $S_0^c,Y_0=0|W=w$ under $P^F$. If only the effect at a single $s_1^\star$ is of interest, then the final assumption above could be replaced by the weaker assumption that, for $P^F$ almost all $w$, $P^F(S_1=s_1^\star|Y_0=0,W=w)=P^F(S_0^c=s_1^\star|Y_0=0,W=w)$. Note that, in either case, the third assumption above is implied by $S_0^c=S_1$ almost surely, but does not require this. As an example of the greater generality of the above assumption, it allows for a true underlying biomarker to be measured with error, where the conditional probability statement is then interpreted for the noised measurement of the underlying biomarker. Note also that this assumption becomes weaker as the baseline covariate $W$ become more predictive of $S_1$ and $S_0^c$. Therefore, using the richest possible baseline covariate by including all available subject-level information is expected to yield the most plausible crossover assumption.

Throughout we also require the strong positivity assumption that, for some $\delta>0$, the treatment mechanism $P(A=1|W)$ falls between $\delta$ and $1-\delta$ with probability one over $W\sim P$. We will also assume that there exists some fixed $\delta>0$ so that any estimate of the treatment mechanism discussed in this paper satisfies the strong positivity assumption with probability approaching one. In Appendix~\ref{app:ident}, we prove the following identifiability result.
\begin{theorem}\label{thm:identifiability}
If \ref{it:sutvacons}, \ref{it:ignorable}, and \ref{it:crossover} hold, then:
\begin{enumerate}
	\item\label{it:margS1} $P^F(S_1=s_1^\star) = \E[P(S=s_1^\star|A=1,W)]$,
	\item\label{it:Y1andS1} $P^F(Y_1=1,S_1=s_1^\star) = \E\left[P(Y=1,S=s_1^\star|A=1,W) \right]$, and
	\item\label{it:notY0andS1} $P^F(Y_0=0,S_1=s_1^\star) = \E\left[P\left(Y=0,S^c=s_1^\star\middle|A=0,W\right)\right]$.
\end{enumerate}
\end{theorem}
These identifiability results enable the identification of $P^F(Y_1=1|S_1=s_1^\star) = \frac{P^F(Y_1=1,S_1=s_1^\star)}{P^F(S_1=s_1^\star)}$ and of $P^F(Y_0=1|S_1=s_1^\star) = 1-\frac{P^F(Y_0=0,S_1=s_1^\star)}{P^F(S_1=s_1^\star)}$, and therefore also enable the identification of the principally stratified relative risk or any other contrast of these two quantities.

It will be convenient to define parameters mapping from our statistical model $\mathcal{M}$, i.e. the nonparametric model that at most places restrictions on the probability of receiving treatment $A=1$ given covariates $W$, to the real line. In particular, for an arbitrary distribution $P'\in\mathcal{M}$, we define parameters corresponding to each of the above identifiability results, where to ease notation we omit the dependence of these parameters on the choice of $s_1^\star$:
\begin{enumerate}
	\item $\Psi_{1}(P')\equiv \E_{P'}[P'(S=s_1^\star|A=1,W)]$,
	\item $\Psi_{2}(P')\equiv \E_{P'}\left[\E_{P'}\left[Y\middle|A=1,S=s_1^\star,W\right]P'(S=s_1^\star|A=1,W) \right]$, and
	\item $\Psi_{3}(P')\equiv \E_{P'}\left[P'\left(S^c=s_1^\star\middle|A=0,Y=0,W\right)P'(Y=0|A=0,W)\right]$.
\end{enumerate}
Note that, under our identifiability conditions, the principally stratified relative risk is equal to
\begin{align}
\RR^P(s_1^\star)&\equiv \frac{\Psi_2(P')}{\Psi_1(P')-\Psi_3(P')}. \label{eq:RRobs}
\end{align}
While the numerator above is bounded between 0 and 1, the denominator may be negative. Therefore, it is possible that the above, which is supposedly identified with a counterfactual relative risk, is negative. This is of course impossible, and therefore indicates a failure of our identifiability assumptions. In Theorem~\ref{thm:falsifiability} of Appendix~\ref{app:falsifiability} we show that, under \ref{it:sutvacons} and \ref{it:ignorable}, the assumption \ref{it:crossover} is not falsifiable from the observed data distribution if $\Psi_4(P)=0$, where
\begin{align*}
\Psi_4(P')\equiv \int \int \Bigl(P'\left(S^c=s_1,Y=0\middle|A=0,w\right) - P'(S=s_1|A=1,w)\Bigr)^+ dP'(S=s_1|A=1,w)dP'(w).
\end{align*}
Above $x^+$ denotes the positive part of $x$. 
If, on the other hand, \ref{it:sutvacons} and \ref{it:ignorable} hold and $\Psi_4^F(P^F)>0$, then it is easy to show that \ref{it:crossover} cannot hold. Our statistical inference for the principally stratified relative risk will hinge on \ref{it:crossover}. The fact that $\Psi_4(P)>0$ implies that \ref{it:crossover} is false suggests that one could provide a test of \ref{it:crossover}, and perform this test as a sanity check before proceeding with inference for the principally stratified relative risk. We do not further consider such a test here.

\section{Estimation for discrete biomarker}\label{sec:estDisc}
Until otherwise stated, we work at fixed $s_1^\star$. Suppose we observe $n$ i.i.d. samples $O_1,\ldots,O_n\sim P$, where $O_i=(W_i,A_i,S_i,Y_i,S_i^c)$. The notational overload on subscripts (indexing subject $i$ and counterfactuals) will not be problematic because the remainder of this work focuses only on observed data quantities.

We will present an estimator of $\mb{\Psi}(P)\equiv (\Psi_k(P) : k=1,2,3)\equiv (\psi_k : k=1,2,3)$ and show that it has a normal limiting distribution. This will enable both a test of \ref{it:crossover} and the construction of confidence intervals for the quantity $\RR^P(s_1^\star)$ that is, under conditions discussed in the previous section, identified with the principally stratified relative risk. Before the presentation of this estimator, whose construction is somewhat involved, Section~\ref{sec:motivation} motivates the theoretical development needed to define the estimator. Section~\ref{sec:estimation} presents the estimator.

\subsection{Motivation}\label{sec:motivation}
We will derive the efficient influence function $\mb{D}^{P'} : \mathcal{O}\rightarrow\mathbb{R}^3$ of the parameter $\mb{\Psi}$, where we note that, as the notation suggests, the influence function depends on the distribution $P'$ at which $\mb{\Psi}$ is evaluated. While a more thorough review of the definition of efficient influence functions is given in Appendix~\ref{app:EIF}, the key result is that they yield a first-order expansion of the form:
\begin{align}
\mb{\Psi}(P')-\mb{\Psi}(P)&= -\E_P\left[\mb{D}^{P'}(O)\right] + \Rem(P,P'), \label{eq:firstord}
\end{align}
where above $\Rem(P,P')\in\mathbb{R}^3$ is small (in Euclidean norm) relative to the leading term on the right whenever $P'$ is close to $P$. In our setting, $P'$ will denote an estimate of $P$, where one really only needs to estimate the conditional expectations and probabilities under $P$ needed to evaluate $\mb{\Psi}$ and $\mb{D}$, rather than the entire distribution. Once this result has been established, we will use it to develop a targeted minimum loss-based estimator (TMLE) of the three-dimensional estimand $\boldsymbol{\psi}$. We can then estimate any smooth function of this quantity, e.g. a relative risk.

Before presenting the TMLE, which is somewhat complex, we provide intuition as to its strong theoretical properties. Suppose we have estimated the components of $P$ needed to evaluate $\mb{\Psi}$ and $\mb{D}$. Denote the estimate by $\hat{P}$. Further suppose that this initial estimate is close to $P$ in the sense that $\Rem(P,\hat{P})\approx 0$, where we later give an explicit expression for this remainder and also give an exact rate requirement to make the approximation symbol precise. Then, by \eqref{eq:firstord}, one has that $\mb{\Psi}(\hat{P})-\mb{\Psi}(P)\approx -\E_P[\mb{D}^{\hat{P}}(O)]$. Because $\mb{D}^{P}$ is mean zero when applied to draws of $O\sim P$, if $\hat{P}$ is close to the true data generating distribution then it is expected that the right-hand side is close to zero. Nonetheless, when $W$ is continuous it is difficult to quantify how close to zero this expectation is. The TMLE overcomes this challenge by obtaining an estimate $\hat{P}^*$ that is a slight fluctuation of $\hat{P}$, where this fluctuated estimate now satisfies
\begin{align}
\frac{1}{n}\sum_{i=1}^n \mb{D}^{\hat{P}^*}(O_i) = \mb{0}. \label{eq:tmleeqn}
\end{align}
In fact, it would suffice for the left-hand side to converge to zero in probability faster than the inverse of the square root of the sample size. 
Arguments from M-estimation can be used to show that the fluctuation step does not destroy the convergence properties of the initial estimate $\hat{P}$, so that $\Rem(P,\hat{P}^*)\approx 0$ whenever this statement holds for the initial, unfluctuated estimator. Combining the above with \eqref{eq:firstord}, we have that
\begin{align*}
\mb{\Psi}(\hat{P}^*)-\mb{\Psi}(P)&\approx \frac{1}{n}\sum_{i=1}^n \left\{\mb{D}^{\hat{P}^*}(O_i) - \E_P\left[\mb{D}^{\hat{P}^*}(O)\right]\right\}.
\end{align*}
Multiplying both sides by the square root of sample size, and noting that, under some minor regularity conditions, $\mb{D}^{\hat{P}^*}$ can be replaced by its mean-square limit (often, though not always, this will be $\mb{D}^P$), we obtain the convergence in distribution result $n^{1/2}\left[\mb{\Psi}(P')-\mb{\Psi}(P)\right]\overset{d}{\longrightarrow} N(\mb{0},\Sigma)$, where the covariance matrix $\Sigma$ can be estimated using the empirical covariance of $\mb{D}^{\hat{P}^*}(O)$. Confidence intervals for \eqref{eq:RRobs} can be derived via the delta method.

Having presented the high-level arguments for our TMLE, we devote the remainder of this section to formally establishing the validity of this estimator. First, we derive the efficient influence functions of the three parameters of interest. Next, we present a TMLE that fluctuates an initial estimate of $P$ so that \eqref{eq:tmleeqn} is satisfied. We then present a theorem giving regularity conditions for the convergence in distribution to a multivariate normal, and we also observe that these regularity conditions can be weakened via cross-validation. Finally, we combine a log transformation with the delta method to demonstrate how to construct confidence intervals for contrasts between the treatment-specific principally stratified risks.

\subsection{Presentation of estimator}\label{sec:estimation}
Per the above Motivation section, the efficient influence function of $\mb{\Psi}$ should play a major role in estimation. Therefore, the following result, which presents an expression for the efficient influence function, will be useful. Before presenting the result, we note that throughout we will use the following conventions for any distribution $P'$, function $f$, and realizations $(a,w)$ of $(A,W)$: $\E_{P'}[f(O)|a,w]=\E_{P'}[f(O)|A=a,W=w]$ and $P'(a|w)= P'(A=a|W=w)$.
\begin{theorem}\label{thm:eif}
Within the model $\mathcal{M}$ that at most places restrictions on the probability of receiving treatment given baseline covariates, the parameter $\mb{\Psi}$ has efficient influence function $\mb{D}^{P'}\equiv (D_k^{P'} : k=1,2,3)$ at $P'$, where, for $k=1,2,3$,
\begin{align*}
D_k^{P'}(o)&\equiv \frac{\indicator{a=a_k}}{P'(a|w)}\left\{f_k(o) - \E_{P'}[f_k(O)|a_k,w] \right\} + \E_{P'}[f_k(o)|a_k,w] -  \Psi_k(P').
\end{align*}
Above, $f_1(o)\equiv \Ind_{S=s_1^\star}$ and $a_1=1$; $f_2(o)\equiv \Ind_{Y=1,S=s_1^\star}$ and $a_2=1$; and $f_3(o)\equiv \Ind_{Y=0,S^c=s_1^\star}$ and $a_3=0$.
\end{theorem}
The proof of this result is in Appendix~\ref{app:est}. It easy to verify that the remainder $\Rem(P,P')\equiv (\Rem_k(P,P') : k=1,2,3)$ in \eqref{eq:firstord} takes the following form:
\begin{align}
\Rem_k(P,P')&= \E_P\left[\left(1-\frac{P(a_k|W)}{P'(a_k|W)}\right)\left(\E_{P'}[f_k(O)|a_k,W]-\E_P[f_k(O)|a_k,W]\right)\right]. \label{eq:rem}
\end{align}
Using this efficient influence function, it remains to develop an estimator satisfying \eqref{eq:tmleeqn}. This estimator is given in Algorithm~\ref{alg:tmle}, which independently estimates each parameter $\Psi_k(P)$, $k=1,2,3$, by invoking a TMLE for estimating $\E\E[f(O)|A=a,W]$ for an arbitrary function $f$ and treatment $a$. Many variants of this $k$-specific TMLE have been presented elsewhere \citep[e.g.,][]{vanderLaan&Rose11}. We note here that we are in fact implementing three separate TMLEs, and so it is not immediately obvious that there exists a single distribution $\hat{P}^*$ such that our estimator is equal to $\Psi(\hat{P}^*)$. We can show that a unique distribution exists if $\hat{P}_2^*(f_2(O)=1|a_2,w)\le \hat{P}_1^*(f_1(O)=1|a_1,w)$ for all $w$.

\begin{algorithm}
\caption{TMLE for Estimating $\Psi_1(P)$, $\Psi_2(P)$, and $\Psi_3(P)$}\label{alg:tmle}
\begin{algorithmic}
\Statex Takes as input $n$ observations $\textnormal{Obs}\equiv\{O_i : i=1,\ldots,n\}$.
\Function{TMLE}{$\textnormal{Obs}$}
	\State \parbox[t]{\dimexpr\linewidth-\algorithmicindent}{\textbf{Initial Estimates:} Define an initial estimator $\hat{P}$ of $P$:
\begin{itemize}
	\item The marginal distribution of $W$ under $\hat{P}$ should be the empirical.
	\item Only the estimates of the components of $P(O|W)$ needed to evaluate $\Psi_k$ and $D_k$, $k=1,2,3$, are needed: $\hat{P}(A=1|W=\cdot)$, $\hat{P}(Y=1|S=s_1^\star,A=1,W=\cdot)$, $\hat{P}(S=s_1^\star|A=1,W=\cdot)$, and $\hat{P}(Y=0,S^c=s_1^\star|A=0,W=\cdot)$.
\end{itemize}
	\Comment{Conditional probability estimates should fall in $(0,1)$.}}
	\For{$k=1,2,3$}\Comment{Recall the definitions of $f_k$ and $a_k$ from Theorem~\ref{thm:eif}.}
		\State \parbox[t]{\dimexpr\linewidth-\algorithmicindent-\algorithmicindent}{\textbf{Fluctuation for Targeting Step:} Using observations $i=1,\ldots,n$, fit the intercept $\hat{\epsilon}_k$ using an intercept-only logistic regression with outcome $f_k(O_i)$, offset $\logit \hat{P}\left\{f_k(O)=1|a,w\right\}$, and weights $\frac{\Ind_{A_i=a_k}}{\hat{P}(A_i|W_i)}$.}
		\State \parbox[t]{\dimexpr\linewidth-\algorithmicindent-\algorithmicindent}{\textbf{Targeting Step:} We now define a fluctuation of $\hat{P}_k^*$ of $\hat{P}$, targeted towards estimation of $\psi_k$. Define 
		\begin{align*}
		\hat{P}_k^*\left\{f_k(O)=1|a_k,w\right\}&\equiv \expit\left\{\logit \hat{P}\left\{f_k(O)=1|a_k,w\right\} + \hat{\epsilon}_k\right\}.
		\end{align*}
		The marginal distribution of $W\sim \hat{P}$ does not need to be fluctuated.
			}
		\State \parbox[t]{\dimexpr\linewidth-\algorithmicindent-\algorithmicindent}{\textbf{Plug-In Estimator:} The estimator of $\psi_k$ is $\hat{\psi_k}\equiv \Psi_k(\hat{P}_k^*)$, and the estimate of the corresponding component of the influence function is $\hat{D}_k\equiv D_k^{\hat{P}_k^*}$.
		}
	\EndFor
	\State \parbox[t]{\dimexpr\linewidth-\algorithmicindent}{\textbf{return} Estimate $\boldsymbol{\hat{\psi}}\equiv (\hat{\psi}_k : k=1,2,3)$ and estimated influence function $\mb{\hat{D}}\equiv (\hat{D}_k : k=1,2,3)$.}
\EndFunction
\end{algorithmic}
\end{algorithm}

We now present the limiting distribution of the proposed estimator. This result relies on the following conditions:
\begin{itemize}
	\item $\mb{\hat{D}}$ belongs to a fixed Donsker class \citep{vanderVaartWellner1996} with probability approaching one;
	\item $\int [\mb{\hat{D}}(o)-\mb{D}(o)]^2 dP(o)$ converges to zero in probability as $n\rightarrow\infty$;
	\item all remainders are negligible: $\max_k \Rem_k(P,\hat{P}_k^*)=o_P(n^{-1/2})$.
\end{itemize}
\begin{theorem}\label{thm:al}
Under the above listed conditions, there exists a covariance matrix $\Sigma$ such that
\begin{align*}
n^{1/2}\left[\boldsymbol{\hat{\psi}}-\boldsymbol{\psi}\right]&= n^{-1/2}\sum_{i=1}^n\mb{D}^{P}(O_i) + o_P(1)\overset{d}{\longrightarrow} \textnormal{Normal}(\mb{0},\Sigma),
\end{align*}
\end{theorem}
The proof of this result is omitted because very similar proofs have already appeared in the literature many times \citep[e.g.,][for overviews]{vanderLaan&Rose11,Kennedy2016}. The limiting covariance matrix $\Sigma$ can be consistently estimated using the empirical covariance of $\mb{D}^{\hat{P}_k^*}(O)$.

To develop sufficient conditions for $\max_k \Rem_k(P,\hat{P}_k^*)$ to be negligible, one can use that each $\Rem_k(P,\hat{P}_k^*)$ is upper bounded by a constant times the product of the root-mean-square distance between $\hat{P}_k^*(a_k|W=\cdot)$ and $P(a_k|W=\cdot)$ and the root-mean-square distance between $\hat{P}_k^*(f_k(O)=1|a_k,W=\cdot)$ and $P(f_k(O)=1|a_k,W=\cdot)$. In a parametric model, each of these root-mean-square distances is $O_P(n^{-1/2})$, and outside of a parametric model sufficient smoothness will enable several estimators to ensure that the product of these two rates is $o_P(n^{-1/2})$. Examples include generalized additive models \citep{Hastie&Tibshirani90,Horowitz2009} and the highly adaptive lasso estimator \citep{vanderLaan2017}. If the data are generated from a randomized clinical trial, and the known treatment assignment probabilities are used, then each $\Rem_k(P,\hat{P}_k^*)$ is exactly zero.

The Donsker condition in the theorem above restricts the flexibility of the initial distribution estimate $\hat{P}$. This restriction may be unpleasant in situations where the dependence of $P\left\{f_k(O)=1|a,w\right\}$ on $w$ is believed to be complicated. In these cases, one can modify the above estimator via a $V$-fold sample splitting approach, resulting in a cross-validated TMLE \citep{Zheng&vanderLaan11}. For ease of exposition, we focus on the case that $V=10$. Here, one partitions the data into ten mutually exclusive and exhaustive folds of approximately equal size. One then fits ten initial estimates of $\hat{P}$, where each initial estimate leaves out a different one of these ten folds. For each observation $i$, let $v(i)$ denote the fold that contains $i$. Let $\hat{P}^{v}$ denote the initial estimate based on the nine out of ten folds that do not include fold $v$. The sample splitting procedure then replaces each evaluation of (a conditional distribution of) $\hat{P}$ at $O_i$ within the for loop over $k$ by $\hat{P}^{v(i)}$, where we recall that $k$ indexes the three dimensions of the output of the parameter $\mb{\Psi}$. A single, non-fold-specific fluctuation $\hat{\epsilon}_k$ is still returned from the logistic regression in this for loop. The fluctuated initial estimates are still fold-specific, with
\begin{align*}
\hat{P}_k^{v,*}\{f_k(O)=1|a,w\}&\equiv \expit\left\{\logit \hat{P}^v\left\{f_k(O)=1|a,w\right\} + \hat{\epsilon}_k\right\}.
\end{align*}
The estimator of $\psi_k$ is $\hat{\psi}_k^{CV}\equiv \frac{1}{n}\sum_{i=1}^n \Psi_k(\hat{P}_k^{v(i),*})$. The estimator $\boldsymbol{\hat{\psi}}^{CV}$ should have the same multivariate normal limit as $\boldsymbol{\hat{\psi}}$ would have if the Donsker condition is satisfied, but this cross-validated TMLE will still have this multivariate normal limit even if this Donsker condition fails.

As discussed in the introduction, the delta method can be used to derive confidence intervals for any smooth mapping $\Gamma : \mathbb{R}^3\rightarrow\mathbb{R}$ of $\boldsymbol{\psi}$, where the efficient influence function is then $\langle\mb{D}^P,\dot{\Gamma}(\boldsymbol{\psi})\rangle$, with $\dot{\Gamma}$ the gradient of $\Gamma$. For a discrete biomarker, as is studied in this section, this will require that the biomarker takes on the value $s_1^\star$ with positive probability, i.e. that $\E[P(S=s_1^\star|A=1,W)]>0$. For example, for the log of the relative risk defined in \eqref{eq:RRobs}, $\Gamma(\mb{x})\equiv \log(x_2)-\log(x_1-x_3)$ and $\dot{\Gamma}(\mb{x}) = \left(-[x_1-x_3]^{-1},x_2^{-1},[x_1-x_3]^{-1}\right)$. For the risk difference on treatment $A=1$ versus on treatment $A=0$, $\Gamma(\mb{x})=(x_2-x_3)/x_1$ and $\dot{\Gamma}(\mb{x}) = x_1^{-2}(x_3-x_2,x_1,-x_1)$. In either case, the influence function of our estimator $\boldsymbol{\hat{\psi}}$ can be estimated using $\langle\mb{\hat{D}},\dot{\Gamma}(\boldsymbol{\hat{\psi}})\rangle$. Letting $\hat{\sigma}^2$ denote the empirical variance of $\langle\mb{\hat{D}}(O),\dot{\Gamma}(\boldsymbol{\hat{\psi}})\rangle$, a 95\% confidence interval for $\Gamma(\boldsymbol{\psi})$ is given by $\Gamma(\boldsymbol{\hat{\psi}})\pm 1.96n^{-1/2}\hat{\sigma}$.

\section{Estimation Under Two-Phase Sampling}\label{sec:2p}
In some settings, such as in many vaccine studies, the post-treatment biomarker $S$ and the post-crossover biomarker $S^c$ will only be measured on a subset of the subjects in the study. For example, the biomarker $S$ may be measured on all treated ($A=1$) cases ($Y=1$) and all members of a prespecified cohort randomly drawn from the set of all subjects, whereas the crossover biomarker $S^c$ may be measured on a random subset of untreated ($A=0$) controls ($Y=0$). In this section, we study estimation in the setting where the biomarkers are discrete.

Section~\ref{sec:2poverview} gives an overview of the notation and assumptions needed in this setting. Section~\ref{sec:2pipwtmle} presents an inverse probability weighted TMLE based on the framework of \cite{Rose2011}. This estimator is appealing because it only requires minor modifications to the estimator from Section~\ref{sec:estimation}, namely in the inclusion of sampling weights. When the baseline covariates are not discrete, the inverse probability weighted TMLE presented in this section may not be efficient, because it does not fully leverage the predictive power of the baseline covariates for the missing biomarker, which is especially important for subjects for whom these biomarkers are not measured. Though \cite{Rose2011} also presented a TMLE that can attain efficiency in settings where the baseline covariates are continuous, we instead present a one-step estimator for these settings, that has the advantage of being easier to describe and implement, but the disadvantage of not being a plug-in estimator.

\subsection{Overview}\label{sec:2poverview}
For treated subjects, let $\Delta$ denote the indicator of having the biomarker $S$ measured. For untreated subjects, let $\Delta$ denote the indicator of having the crossover biomarker $S^c$ measured. We use the convention that $\Delta$ is always one for untreated cases, which does not cause a problem because $S^c$ is defined to be zero for all untreated cases.

We make the missing at random assumptions that 
the indicator $\Delta$ is independent of the biomarkers $(S,S^c)$ given phase-one information $(W,A,Y)$, where we recall that $S$ is degenerate for untreated subjects and $S^c$ is degenerate for treated subjects and untreated controls. For each subject $i$, the indicator of having the biomarker measured is given by $\Delta_i$. We suppose that $\Delta_1,\ldots,\Delta_n$ is an i.i.d. sequence given $(O_1,\ldots,O_n)$, where we remind the reader that $O_i\equiv (W_i,A_i,S_i,Y_i,S_i^c)$. We let $\tilde{O}\equiv (W,A,\Delta S,Y,\Delta S^c)$ denote the censored data structure, which is equal to $O$ except that the biomarkers are unobserved for all subjects with $\Delta=0$. Let $\tilde{P}$ denote the corresponding distribution of $(\Delta,\tilde{O})$. Our estimation scheme will be a function of the observed $n$ i.i.d. draws of $\tilde{O}\sim\tilde{P}$. 

Let $\Pi(\tilde{O})$ denote the probability that $\Delta=1$ given $\tilde{O}$. For brevity, we let $\pi\equiv \Pi(\tilde{O})$ for an arbitrary $\tilde{O}\sim \tilde{P}$ and, for each subject $i$, we let $\pi_i\equiv \Pi(O_i)$. We note that $\pi_i=1$ for all untreated cases. Furthermore, in a case-cohort sampling scheme, $\pi_i=1$ for all treated cases.

\subsection{Inverse Probability Weighted TMLE}\label{sec:2pipwtmle}
In this section, we will assume that $\pi_i$ is known for each subject. If $\pi_i$ only relies on $W$ through a discrete coarsening $V$ of $W$, then the inference that we propose will be valid, albeit conservative, if $\pi_i$ is replaced by a nonparametric maximum likelihood estimator $\hat{\pi}_i$ of the probability that $\Delta=1$ conditional on $(V,A,Y)$. In fact, using this nonparametric maximum likelihood estimator typically leads to asymptotic efficiency gains even in our setting where $\pi_i$ is known, and never leads to a loss of asymptotic efficiency in this setting \citep[Theorem~2.3 of][]{vdL02}.

Defining the inverse probability weighted (IPW) TMLE requires only minor modifications to Algorithm~\ref{alg:tmle}. The modifications to our estimation scheme hinge on the fact that, for any $f$, $\E_P[f(O)]=\E_{\tilde{P}}[\Delta\pi^{-1}f(\tilde{O})]$. This fact is useful both when $f$ is a loss function, in which case it suggests to include weights $\Delta\pi^{-1}$ in the estimation procedure, and when $f$ is an influence function, in which case it suggests to weight $\mb{D}^P$ in the leading term on the right-hand side of \eqref{eq:firstord}. The estimation scheme, which represents a two-phase sampling TMLE as presented in \cite{Rose2011}, is presented in Algorithm~\ref{alg:2phasetmle} in Section~\ref{sec:2p}. When obtaining the initial estimate $\hat{P}$ of $P$, one can use loss-based approaches, e.g. minimizing the empirical mean-squared error or Kullback-Leibler divergence to estimate the needed components of $P$. The key in the two-phase sampling case is to include the inverse of the (estimated) subject-level probabilities for belonging to the second phase as weights, i.e. including as weights $\Delta_i\pi_i^{-1}$, in whatever loss-based estimation procedure that is used to estimate the needed components of $P$ \citep[see Section~3.2.2 of][]{Rose2011}.

Confidence intervals for the IPW TMLE can again be constructed using that, under regularity conditions closely related to those of Theorem~\ref{thm:al}, $n^{1/2}[\boldsymbol{\hat{\psi}}-\boldsymbol{\psi}]\overset{d}{\longrightarrow} \textnormal{Normal}(\mb{0},\tilde{\Sigma})$. The covariance matrix $\tilde{\Sigma}$ can be estimated by the empirical covariance of the inverse weighted influence function $\left(\Delta_i \bar{\pi}_i^{-1} \hat{D}_k(\tilde{O}_i) : k=1,2,3\right)$, where $\bar{\pi}_i$ is defined in Algorithm~\ref{alg:2phasetmle} as a rescaling of $\pi_i$ that ensures that $\Delta\bar{\pi}^{-1}$ has empirical mean one within both the treated and untreated strata of the observations. The only difference relative to the estimation of $\Sigma$ is that the influence function is now weighted by $\Delta_i \bar{\pi}_i^{-1}$. Confidence intervals for the relative risk can again be generated using the delta method.

\subsection{One-Step Estimator that Fully Leverages Continuous Baseline Covariates}\label{sec:2pee}
We now present a one-step estimator that fully leverages the predictive power of baseline covariates, treatment status, and endpoint status for the missing values of the biomarker. Under some regularity conditions, this estimator is efficient in the model in which the probability of treatment given covariates is unknown, and is only slightly inefficient in settings where the probability of treatment given covariates is known \citep{Marsh2016}. In the setting where the treatment mechanism is unknown, the efficient influence function is given by
\begin{align*}
\widetilde{\mb{D}}^{P'}(\tilde{o})&\equiv \delta \pi^{-1}\mb{D}^{P'}(\tilde{o}) + \left(1-\delta \pi^{-1}\right)\E_{P'}\left[\mb{D}^{P'}(\tilde{O})\middle|\Delta=1,a,w,y\right],
\end{align*}
where $\tilde{o}\equiv (w,a,\delta s,y,\delta s^c)$ and we abuse notation and let $\pi=\Pi(\tilde{o})$. The elements of the above vector rewrite as
\begin{align}
\widetilde{D}_k^{P'}(\tilde{o})\equiv\,& \frac{\indicator{a=a_k}}{P'(a|w)}\left\{\delta \pi^{-1}f_k(o) + \left(1-\delta \pi^{-1}\right)\E_{P'}\left[f_k(O)\middle|\Delta=1,a,w,y\right] - \E_{P'}[f_k(O)|a_k,w] \right\} \nonumber \\
&+ \E_{P'}[f_k(o)|a_k,w] -  \Psi_k(P'), \label{eq:Dtildedef}
\end{align}
To define our estimator, we first suppose that we have an estimator $\hat{P}$ of the distribution $P$ of the data that would have been observed had the probability of membership in the second phase of sampling been equal to one for all subjects. This estimator can be obtained via loss-based learning using inverse probability weighting \citep[see Section~3.2 of][]{Rose2011}. Next, we define our estimator as $\boldsymbol{\hat{\psi}}\equiv \mb{\Psi}(\hat{P}) +  \frac{1}{n}\sum_{i=1}^n \widetilde{\mb{D}}^{P'}(\tilde{O}_i)$. Let $\hat{\Sigma}_n$ denote the empirical covariance matrix of $\widetilde{\mb{D}}^{P'}(\tilde{O}_i)$, $i=1,\ldots,n$, where here we acknowledge the notational overload with the covariance matrix $\hat{\Sigma}_n$ from Section~\ref{sec:estimation} (single-phase sampling). Under some conditions, one can show that $n^{1/2}\hat{\Sigma}_n^{-1/2}\left[\boldsymbol{\hat{\psi}}-\mb{\Psi}(P)\right]\overset{d}{\longrightarrow} N(\mb{0},\Id)$, where $\Id$ denotes the $3\times 3$ identity matrix. One can then use the delta method to develop confidence intervals for any contrast based on the principally stratified means $\boldsymbol{\psi}$, such as the log relative risk.

To gain some intuition on this estimator, note that, if the baseline covariates are strongly predictive of the biomarker among treated subjects and of the crossover biomarker among uninfected untreated subjects, then the estimate of $\E_P\left[f_k(O)\middle|\Delta=1,a,w,y\right]$ will be approximately equal to $f_k(o)$. Therefore, $\delta \pi^{-1}f_k(o) + \left(1-\delta \pi^{-1}\right)\E_{P'}\left[f_k(O)\middle|\Delta=1,a,w,y\right]$ is approximately equal to $f_k(o)$, so that $\widetilde{D}_k^{P'}(\tilde{o})$ is approximately equal to the efficient influence function in the setting where phase two information is observed on everyone.

Finally, we note that, though we have assumed that $\Pi$ is known, the proposed one-step estimator will also yield valid inference if $\Pi$ is replaced by an estimate and a doubly robust term is negligible.

\section{Continuous Biomarker}\label{sec:cont}
\subsection{Algorithm and Theoretical Guarantees}
Thus far we have assumed that the biomarker is discrete. Suppose now that the biomarkers $S_1,S_0^c$ are continuous, with support in $\mathbb{R}$. We will show the methods that we have proposed immediately extend to this case once one replaces the indicators that $S=s_1^\star$ and $S^c=s_1^\star$ by kernels. The extension to the case that the biomarkers have support in $\mathbb{R}^d$, $d>1$, is straightforward by using kernels for $\mathbb{R}^d$-valued data, though, given that kernel smoothers are highly susceptible to the curse of dimensionality, we expect that this method will only yield informative estimates when $d$ is small. For simplicity, we consider the full sampling case where all treated subjects and all untreated subjects without the event have their biomarker measured. As in the previous section, our objective will be to estimate the effect of treatment conditional on a biomarker value $s_1^\star$.

This treatment effect will be a function of three parameters, which are defined with respect to the (conditional) Lebesgue densities of the biomarkers $S$ and $S^c$, which are assumed to exist. To avoid introducing new notation for these densities, we will abuse notation so that writing ``$S=s$'' or ``$S^c=s$'' in a conditional probability statement refers to the conditional density of $S$ or $S^c$ at $s$. Once we have made this abuse of notation, the definition of $\mb{\Psi}(P)$ from Section~\ref{sec:prelim} does not need to be changed to study the setting where the biomarkers are continuous. So, for example, $\Psi_1(P')\equiv \E_{P'}[P'(S=s_1^\star|A=1,W)]$ represents the Lebesgue density of $S$ at $s_1^\star$, conditional on $A=1$ and $W$, averaged across values of $W\sim P'$. 

Let $K$ denote a kernel in $\mathbb{R}$, where this function is centered at zero and integrates to one. Examples of kernels include $K(x)=\Ind\{-1\le 2x\le 1\}$ (uniform kernel), $K(x) \propto \exp(-x^2/2)$ (Gaussian kernel), and $K(x) \propto \frac{1}{2}(3-x^2)\exp(-x^2/2)$ (fourth-order Gaussian kernel), where the constants in the `$\propto$' statements are chosen so that the kernels integrate to one. The first two kernels are second-order in the sense that $\int x^t K(x) dx = 0$ for $t=1$, whereas the final kernel is fourth-order in the sense that $\int x^t K(x) dx$ for $t=1,2,3$. Generally, a kernel is of order $r$ if $\int x^t K(x) dx=0$ for all $t=1,\ldots,r-1$. We define $K_{h}$ as the function $x\mapsto h^{-1} K[xh^{-1}]$, where $h>0$ is the bandwidth.

Before presenting our algorithm, we introduce the kernel-dependent pseudo-outcomes that we will use to replace $f_k$, $k=1,2,3$. In particular, we define $f_{1,h}(O_i)\equiv K_{h}(S-s_1^\star)$, $f_{2,h}(O_i)\equiv \Ind_{Y=1} K_{h}(S-s_1^\star)$, and $f_{3,h}(O_i)\equiv \Ind_{Y=0} K_{h}(S^c-s_1^\star)$. We continue to use the definitions $a_1=1$, $a_2=1$, and $a_3=0$. For an arbitrary $P'$, we also define $D_{k,h}^{P'}$ analogously to $D_k^{P'}$ from Theorem~\ref{thm:eif}, but with each instance of $f_k$ replaced by $f_{k,h}$. We also define a smoothed version of our parameter. To do this, it will be useful to make the dependence of $\mb{\Psi}(P)$ on $s_1^\star$ explicit in the notation: in particular, we define $\mb{\Psi}(P;s)\equiv \{\Psi_k(P;s) : k=1,2,3\}$ for all $s$ in the support of $S|A=1$ under $P$. We now define and give an equivalent expression for the kernel-smoothed version of $\mb{\Psi}(P;s_1^\star)$:
\begin{align}
\mb{\Psi}_h(P)\equiv \int \mb{\Psi}(P;s) K_h(s-s_1^\star) ds = \E\left[\E\left[f_{k,h}(O)\middle|a_k,W\right]\right]. \label{eq:smoothedparam}
\end{align}
The equivalent expression relies on changing the order of integration between the integral over $s$ and the integral over $W\sim P$. We write the three components of $\mb{\Psi}_h(P)$ as $\Psi_{k,h}(P)$ for $k=1,2,3$. It will be convenient to define $q_{k,h}$ as the function $w\mapsto \E\left[f_{k,h}(O)\middle|a_k,w\right]$.

Our proposed estimation procedure is presented in Algorithm~\ref{alg:tmlecont}. In our procedure, we estimate $q_{k,h}$ for each $k=1,2,3$, and so by the latter identity in \eqref{eq:smoothedparam} it follows that we can use these expectations to evaluate $\mb{\Psi}_h$ at our estimated distribution.
\begin{algorithm}
\caption{CV-TMLE for Estimating $\Psi_1(P)$, $\Psi_2(P)$, and $\Psi_3(P)$ when the biomarker is continuous}\label{alg:tmlecont}
\begin{algorithmic}
\Statex Takes as input $n$ observations $\textnormal{Obs}\equiv\{O_i : i=1,\ldots,n\}$ and a kernel $K_{h}$.
\Statex For simplicity, focus on 10-fold cross-validation.
\Function{CV-TMLE}{$\textnormal{Obs}$}
	\For{$v=1,\ldots,10$}
	\State \parbox[t]{\dimexpr\linewidth-\algorithmicindent-\algorithmicindent}{\textbf{Initial Estimates:} Using data from training set $v$, define an initial estimator $\hat{P}^v$ of $P$:
\begin{itemize}
	\item The marginal distribution of $W$ under $\hat{P}^v$ should be the empirical of the observations in training set $v$.
	\item Using only observations in the training set $v$, compute an estimate $w\mapsto \hat{P}^v(A=1|w)$ of $w\mapsto P^v(A=1|w)$, and an estimate $w\mapsto\hat{q}_{k,h}^v(w)$ of $w\mapsto q_{k,h}(w)$.
\end{itemize}
	\Comment{Each $q_{k,h}$ should have range in $(0,\infty)$.}}
	\EndFor
	\For{$k=1,2,3$}
		\State \parbox[t]{\dimexpr\linewidth-\algorithmicindent-\algorithmicindent}{\textbf{Fluctuation for Targeting Step:} Define $\hat{\epsilon}_k$ as a minimizer (in real-valued $\epsilon$) of
		\begin{align*}
		\left(\frac{1}{n}\sum_{i=1}^n \frac{\Ind_{A_i=a_k}}{\hat{P}^v(A_i|W_i)} \left[f_{k,h}(O_i) - \exp\left\{\log \hat{q}_{k,h}^{v(i)}(W_i) + \epsilon\right\}\right]\right)^2.
		\end{align*}
		}
		\State \parbox[t]{\dimexpr\linewidth-\algorithmicindent-\algorithmicindent}{\textbf{Targeting Step:} For each fold $v$, define
		\begin{align*}
		\hat{q}_{k,h}^{v,*}(w)&\equiv \exp\left\{\log \hat{q}_{k,h}^{v}(w) + \hat{\epsilon}_k\right\}.
		\end{align*}
		The marginal distribution of $W\sim \hat{P}$ does not need to be fluctuated.
			}
		\State \parbox[t]{\dimexpr\linewidth-\algorithmicindent-\algorithmicindent}{\textbf{Estimator:} The estimator of $\psi_{k,h}$ is $\frac{1}{n}\sum_{i=1}^n \Psi_{k,h}(\hat{P}_{k,h}^{v(i),*})$, and, for each subject $i$, the estimate of the corresponding component of the smoothed influence function is $\hat{D}_{k,h,i}\equiv D_{k,h}^{\hat{P}_{k,h}^{v(i),*}}$.
		}
	\EndFor
	\State \parbox[t]{\dimexpr\linewidth-\algorithmicindent}{\textbf{return} Estimate $\boldsymbol{\hat{\psi}}_{h}\equiv (\hat{\psi}_{k,h} : k=1,2,3)$ and estimated smoothed influence function $\mb{\hat{D}}_{h,i}\equiv (\hat{D}_{k,h,i} : k=1,2,3)$.}
\EndFunction
\end{algorithmic}
\end{algorithm}

Our estimator can be analyzed using standard arguments for kernel estimators. We will break our analysis into a study of the bias and a study of the variance of our estimator. We will show that our estimator has small bias for the smoothed principally stratified mean vector $\mb{\Psi}_h(P)$ at a user-defined value of $h$. Therefore, the bias of our estimator for $\mb{\Psi}(P)$ will be driven by the convergence of a kernel-smoothed parameter to the true target parameter. As is standard for kernel estimators, the rate of convergence of the bias will be driven by the differentiability of $s\mapsto \mb{\Psi}(P;s)$ at $s$ and by the order of the kernel that we select.
\begin{theorem}[Bias for Unsmoothed Parameter]\label{thm:bias}
Let $K$ be an $r^{\textnormal{th}}$-order kernel and suppose that, for each $k\in\{1,2,3\}$, $s\mapsto \Psi_{k,h}(P;s)$ is $(t+1)^{\textnormal{th}}$-order differentiable on the real line with uniformly bounded $(t+1)^{\textnormal{th}}$ derivative. If $\int |u^{\min\{r,t+1\}} K(u)| du<\infty$, then $\norm{\mb{\Psi}_h(P)-\mb{\Psi}(P)}=O(h^{\min\{r,t+1\}})$.
\end{theorem}
Next, we study the variance of our estimator, which is (asymptotically) unbiased for the parameter evaluated at the kernel-smoothed data generating distribution. Before presenting this result, we will define the covariance matrix $\Sigma_n$ to be used in the upcoming theorem. For $\bar{q}_k$ equal to the limit of the estimates of $q_{k,\hat{h}_n}$ as $n\rightarrow\infty$ (see the upcoming Condition~\ref{it:L2lim}) and $\epsilon\in\mathbb{R}$, we let
\begin{align}
g_{k,\epsilon,\bar{h}_n}(o)&\equiv \frac{\Ind_{a=a_k}}{P(a_k|w)} \left[f_{k,\bar{h}_n}(o) - \exp\left\{\log \bar{q}_k(w) + \epsilon\right\}\right] + \exp\left[\log \bar{q}_k(w) + \epsilon\right] \label{eq:glimdef}
\end{align}
and define $\Sigma_n$ as the $3\times 3$ matrix representing the covariance matrix under $O\sim P$ of $\left(\bar{h}_n^{1/2} g_{k,\bar{\epsilon}_k,\bar{h}_n}(O) : k=1,2,3\right)$, where $\bar{\epsilon}_k$, $k=1,2,3$, is defined in the upcoming Condition~\ref{it:barepsilon}. The following theorem relies on conditions that are stated in the upcoming Section~\ref{sec:conds}.
\begin{theorem}[Estimation of Smoothed Parameter]\label{thm:estsmooth}
If Conditions \ref{it:L2lim} through \ref{it:covmat} hold, then
\begin{align}
(n\hat{h}_n)^{1/2}\Sigma_n^{-1/2}\left[\boldsymbol{\hat{\psi}}_{\hat{h}_n}-\mb{\Psi}_{\hat{h}_n}(P)\right]&\overset{d}{\longrightarrow} N(\mb{0},\Id).\label{eq:triangclt}
\end{align}
\end{theorem}
The proof of this result is given in Appendix~\ref{app:cont}. Under mild conditions, the $\Sigma_n$ matrices will be consistently estimable via the cross-validated empirical covariance matrix $\hat{\Sigma}_n\equiv \frac{1}{n}\sum_{i=1}^n \mb{\hat{D}}_{\hat{h}_n,i} [\mb{\hat{D}}_{\hat{h}_n,i}]^T$.

Note that $\mb{\Psi}_{\hat{h}_n}(P)$ is a data adaptive parameter \citep{vanderLaan&Hubbard13}, in the sense that it is random through the bandwidth $\hat{h}_n$. One of the conditions for the theorem, listed in the next subsection, is that there exists a deterministic sequence $\bar{h}_n$ such that $\hat{h}_n/\bar{h}_n\rightarrow 1 + o_P(\bar{h}_n^{1/2})$. In fact, many selection procedures satisfy the (typically stronger) condition that $\hat{h}_n/\bar{h}_n-1=O_P(n^{-1/2})$ for a deterministic sequence $\bar{h}_n$ \citep{Chiu1991,Halletal1991,Fan&Marron1992}. Under these stronger conditions, the difference between the data adaptive smoothed parameter at bandwidth $\hat{h}_n$ and the sample size dependent, but deterministic, smoothed parameter at bandwidth $\bar{h}_n$ is of the order $n^{-1/2}$, so that one can replace the data adaptive parameter by $\mb{\Psi}_{\bar{h}_n}(P)$ in the above theorem without changing the result.

An estimator of the smoothed relative risk $\RR_h^P(s_1^\star)\equiv \Psi_{2,h}(P)/[\Psi_{3,h}(P)-\Psi_{1,h}(P)]$ is given by $\widehat{\RR}_{\hat{h}_n}(s_1^\star)\equiv \boldsymbol{\hat{\psi}}_{2,\hat{h}_n}/[\boldsymbol{\hat{\psi}}_{3,\hat{h}_n}-\boldsymbol{\hat{\psi}}_{1,\hat{h}_n}]$. Similar delta method arguments used to develop confidence intervals for discrete biomarkers can be used in the continuous setting. Then, Wald-type confidence intervals for the log smoothed relative risk can be defined by using that, under the conditions of Theorem~\ref{thm:estsmooth}, $(n\hat{h}_n)^{1/2}\hat{\sigma}_{\hat{h}_n}^{-1}[\log \widehat{\RR}_{\hat{h}_n}(s_1^\star)-\log \RR_{\hat{h}_n}^P(s_1^\star)]$ converges to a standard normal distribution, where, for $\dot{\Gamma}(\mb{x}) = \left(-x_2^{-1},-[x_1-x_3]^{-1},[x_1-x_3]^{-1}\right)$, we define
\begin{align*}
\hat{\sigma}_{\hat{h}_n}^2\equiv \frac{1}{n}\sum_{i=1}^n \left[\langle\mb{\hat{D}}_{\hat{h}_n,i},\dot{\Gamma}(\boldsymbol{\hat{\psi}}_{\hat{h}_n})\rangle-\frac{1}{n}\sum_{j=1}^n \langle\mb{\hat{D}}_{\hat{h}_n,j},\dot{\Gamma}(\boldsymbol{\hat{\psi}}_{\hat{h}_n})\rangle\right]^2.
\end{align*}

Combining the above two theorems demonstrates that there are rates for $\bar{h}_n$ such that asymptotically normal inference is possible for the unsmoothed parameter $\RR^P(s_1^\star)$. Indeed, balancing the standard error of $\boldsymbol{\hat{\psi}}_{\hat{h}_n}$, i.e. $O((n\hat{h}_n)^{-1/2})$, and its bias for estimating the unsmoothed parameter, i.e. $O(h^{\min\{r,t+1\}})$ for $r,t$ defined in Theorem~\ref{thm:bias}, shows that the bias converges more quickly than the standard error if $h=o(n^{-1/(2\min\{r,t+1\} + 1)})$. Nonetheless, selecting these rates for $\bar{h}_n$ requires undersmoothing, i.e. choosing a rate for $\bar{h}_n$ that is slower than the rate that is optimal according to a criterion such as the mean-squared error of a density estimator at $S=s_1$, conditional on $A=1$. 
Following the estimation of the causal effect of continuous treatments in \cite{Kennedyetal2016}, who himself followed the suggestion of \cite{Wasserman2006}, we focus our inference solely on the smoothed relative risk parameter, rather than on the original, unsmoothed parameter. The rate of decay on the bias from the above theorem is still interesting: it shows that this smoothed parameter is getting close to the unsmoothed parameter as the sample size grows. Nonetheless, we do not attempt to quantify the proximity of the smoothed parameter to the unsmoothed parameter in the confidence interval that we construct. Once one has adopted this perspective, one has flexibility in selecting the bandwidth $h$. For example, one could use a cross-validated mean integrated squared error criterion, where this criterion is selected for estimating the marginal density of $S$ among all treated subjects.

An alternative approach would be to use the recent work of \citep{Bibaut&vanderLaan2017}, that enables principled selection of an undersmoothed bandwidth for non-pathwise-differentiable parameters such as $\mb{\Psi}$. Consideration of the performance of this method for the estimation of $\mb{\Psi}$ is beyond the scope of this work.

\subsection{Conditions for Asymptotically Normal Estimation of the Smoothed Parameter}\label{sec:conds}
Theorem~\ref{thm:estsmooth} used the following conditions to prove the asymptotic normality of the estimator for continuous biomarkers:
\begin{enumerate}[resume=identassumptions,label=(A\arabic*),ref=A\arabic*]
	\item\label{it:L2lim} The treatment mechanism estimates satisfy $\max_v \norm{\hat{P}^v(a|W=\cdot)-P(a|W=\cdot)}=o_P(1)$. Furthermore, for $k=1,2,3$, there exists a function $\bar{q}_k$ such that $\max_v \norm{\hat{q}_{k,\hat{h}_n}^v-\bar{q}_k}=o_P(1)$.
	\item\label{it:band} There exists a deterministic positive sequence $\{\bar{h}_n\}_{n=1}^\infty$ such that $\hat{h}_n/\bar{h}_n\rightarrow 1 + o_P(\bar{h}_n^{1/2})$ and $n\bar{h}_n$ diverges to infinity.
	\item\label{it:rem} $\max_v \norm{\hat{q}_{k,h}^v - q_{k,h}}\norm{\hat{P}^v(a_k|W=\cdot)-P(a_k|W=\cdot)}=o_P(n^{-1/2}\hat{h}_n^{-1/2})$.
	\item\label{it:kern} The kernel $K$ can be evaluated using a composition of a finite number of arithmetic operations ($+$, $-$, $\div$, $\times$), indicator functions ($x\mapsto \Ind\{x>c\}$ for a constant $c$), and exponential functions ($x\mapsto e^x$). Furthermore, the kernel $K$ is bounded and uniformly continuous.
	\item\label{it:estsbounded} There exists a constant $c>0$ such that $P\{c^{-1}<\min_{k,v} \hat{q}_{k,h}^v(W)\le\max_{k,v} \hat{q}_{k,h}^v(W)< c\}=1-o_P(1)$ and a constant $\delta>0$ such that $P\{\max_{k,v,a} 1/\hat{P}_k^v(a|W)< \delta^{-1}\}=o_P(1)$.
	\item\label{it:barepsilon} For each $k\in\{1,2,3\}$, there exists a fixed $\bar{\epsilon}_k\in\mathbb{R}$ such that $\hat{\epsilon}_{k,n}\rightarrow\bar{\epsilon}_k$ in probability as $n\rightarrow\infty$, where here (and here only) we make the dependence of $\hat{\epsilon}_k$ on $n$ explicit in the notation by denoting this quantity by $\hat{\epsilon}_{k,n}$.
	\item\label{it:covmat} There exists some fixed $\delta>0$ such that, with probability approaching one, all of the eigenvalues of $\Sigma_n$ are bounded below by $\delta$.
\end{enumerate}
We now discuss these conditions.

\noindent\textbf{Condition~\ref{it:L2lim}.} This condition requires mean-square consistency of the needed nuisance parameter estimates, where we require this consistency for the estimators within each training fold. If the same estimators are used within each training fold, then this is equivalent to requiring mean-square consistency of the estimation of the needed nuisance parameters on the full sample. This condition does not require any rate for this convergence, and so is a minor assumption relative to \ref{it:rem}, though we note that \ref{it:L2lim} is not necessarily implied by \ref{it:rem} (e.g., if the treatment mechanism is known and one uses this information in estimation). Note that here we have assumed that the treatment mechanism estimates converge to the true treatment mechanism in mean-square, whereas we have only required that $\hat{q}_{k,\hat{h}_n}^v$ has some limit. The condition on the treatment mechanism seems minor given that the treatment mechanism is known in most settings where crossover designs would be employed.

\noindent\textbf{Condition~\ref{it:band}.} This condition requires that the bandwidth selection procedure that is employed will eventually return a (nearly) deterministic result. The procedures of \citep{Chiu1991,Halletal1991} yield data-driven bandwidth selections $\hat{h}_n$ that satisfy $\hat{h}_n/h_n^\diamond = 1+O_P(n^{-1/2})$, where here $h_n^\diamond$ is the optimal bandwidth in terms of mean integrated squared error. In fact, \cite{Fan&Marron1992} demonstrates that the $O_P(n^{-1/2})$ term quantifying the multiplicative deviation of these bandwidth selection procedures from the optimal bandwidth is optimal even in terms of the leading constant. Therefore, if one selects $\hat{h}_n$ to optimize an estimate of the density of $S$ among treated subjects, we expect that the selected bandwidth will eventually behave as a deterministic sequence and therefore that \ref{it:band} will be satisfied.

\noindent\textbf{Condition~\ref{it:rem}.} This condition requires that product of the estimates of the treatment mechanism and the smoothed densities $q_{k,h}$ converge to zero faster than does the standard error of our estimator. If the treatment mechanism is known, as is the case in a randomized trial, then one can ensure that this product is exactly zero at every sample size. Otherwise, the plausibility of this condition will typically rely on the dimensionality of $W$. As a reference point, we note  that, if a parametric model could be correctly specified for an unknown treatment mechanism, then this condition would also automatically be satisfied if $\hat{h}_n\rightarrow 0$ in probability, since in this case the product from this condition would be $O_P(n^{-1/2})$.

\noindent\textbf{Condition~\ref{it:kern}.} This condition does not rely on the data or the data generating distribution and therefore can be verified in applications. We use this condition to ensure certain empirical process and envelope conditions that we use in our proof are satisfied. This condition is satisfied by most kernels used in practice. For example, it is satisfied by $K(x)=\Ind\{-1/2\le x\le 1/2\}$ (uniform kernel), $K(x) = \exp(-x^2/2)$ (Gaussian kernel), and $K(x) = \frac{1}{2}(3-x^2)\exp(-x^2/2)$ (fourth-order Gaussian kernel).

\noindent\textbf{Condition~\ref{it:estsbounded}.} This assumption ensures that our estimates both of the density $q_{k,h}$ and of the treatment mechanism do not fall too close to zero or one, and can be enforced in practice via truncation. For the truncation not to damage the chance of the consistency of these estimates in \ref{it:rem}, Condition~\ref{it:estsbounded} requires that the true treatment mechanism does not fall too close to zero or one. This will be plausible for the treatment mechanism in a randomized trial in which all subjects have a non-negligible probability of receiving or not receiving treatment. The assumption on the density $q_{k,h}$ requires that, within each stratum of covariates, any given realization of the biomarker is not very rare or very likely. Furthermore, the assumption on $q_{k,h}$ requires that the probability of having $Y=1$ within baseline covariate strata of subjects not receiving treatment at baseline is never too common (this is to be expected for rare events), and also that the event probability is bounded away from zero among treated subjects within each stratum of baseline covariates.

\noindent\textbf{Condition~\ref{it:barepsilon}.} The existence of a limit $\bar{\epsilon}_k$ for $\hat{\epsilon}_k$ can be ensured under mild conditions using arguments for Z-estimators \citep{vanderVaartWellner1996}. We omit these arguments here for brevity.

\noindent\textbf{Condition~\ref{it:covmat}.} This condition aims to ensure that the limit is non-degenerate. Its plausibility can be checked in practice by checking the size of the eigenvalues of the estimator $\hat{\Sigma}_n$ of $\Sigma_n$, given by $\hat{\Sigma}_n\equiv \frac{1}{n}\sum_{i=1}^n \mb{\hat{D}}_{\hat{h}_n}(O_i) [\mb{\hat{D}}_{\hat{h}_n}(O_i)]^T$, and ensuring that the smallest eigenvalue is not close to zero.

\section{Simulation}\label{sec:sim}
We conduct simulations similar to the settings in \citet{Gilbert&Hudgens2008}, though with continuous rather than discrete baseline covariates and biomarker. Here we focus on a univariate post-treatment biomarker because such biomarkers are often of interest in vaccine studies \citep[see][for examples]{Corey2015}. We consider sample size 5000 with equal number of people in treatment $A=1$ and control $A=0$. We generate $(W,S_1)$ from bivariate normal distribution with mean $(0.41, 0.41)^T$ and covariance matrix 
\[
\begin{pmatrix}
0.55^2 & 0.55^2 \cdot0.5 \\
0.55^2\cdot0.5 & 0.55^2 \\
\end{pmatrix}.
\]
For $a=0,1$, we generate the potential outcome from the logistic model 
\[
P^F(Y_a=1|W,S_1) = \text{expit}(\beta_0+\beta_1a+\beta_2W+\beta_3S_1+\beta_4aS_1)
\]
where $\text{expit}(x) \equiv \exp(x)/(1+\exp(x))$ and $\beta_0=1.2, \beta_1=-0.6, \beta_2=-0.5, \beta_3=-0.1, \beta_4=-0.9$, such that the average disease rate is $9.5\%$ among subjects receiving $A=1$ and is $18.8\%$ among subjects receiving $A=0$. We consider a full sampling setting where all treated subjects ($A=1$) and all untreated controls ($A=0,Y=0$) have the biomarker measured. We consider Gaussian kernel $K_h(x) = \frac{1}{\sqrt{2\pi}h}\exp(-\frac{x^2}{2h^2})$ to smooth $\indicator{S_1=s_1^*}$. We use \verb=R= package \verb=SuperLearner= \citep{SuperLearnerPackage} to estimate the kernel-smoothed $\mb{\Psi}_h(P)$ in equation (\ref{eq:smoothedparam}). We consider GLM, GLM with interaction, stepwise regression, neural network and sample mean in our SuperLearner library, whose corresponding wrapper functions in the \texttt{SuperLearner} package are respectively given by \texttt{SL.glm}, \texttt{SL.glm.interaction}, \texttt{SL.step}, \texttt{SL.nnet}, \texttt{SL.mean}. All of these functions were run at their default settings in version~2.0-22 of the \texttt{SuperLearner} package. 

We calculate two versions of the true parameters: one without smoothing and one with smoothing. The true parameter without smoothing can be derived from the following identity 
\begin{align*}
\Psi_1(P;s_1^*) & = P^F(S_1=s_1^*)  \\
\Psi_2(P;s_1^*) & = \int P^F(Y_1=1|S_1=s_1^*, W=w)P^F(S_1=s_1^*, W=w) dw \\
\Psi_3(P;s_1^*) & = \int P^F(Y_0=0|S_1=s_1^*, W=w)P^F(S_1=s_1^*, W=w) dw 
\end{align*}
The target parameter is log relative risk $\Psi_{h} = \log(\frac{\Psi_{2,h}}{\Psi_{1,h} - \Psi_{3,h}})$. The smoothed version of the true parameter is calculated by \eqref{eq:smoothedparam}.  
We calculate the $95\%$ Wald-type confidence interval and its coverage over the unsmoothed true parameter and smoothed true parameter across 1000 simulations with different kernel bandwidth $h$. Figure \ref{fig:cov} shows the confidence interval coverage for a fixed $s_1^* = 0.6$ with different choices of $h$. The coverage probability for the unsmoothed vaccine efficacy parameter drops below the nominal level when $h$ is small or large.  The coverage probability for the smoothed parameter drops below the nominal level when $h$ is small but does not drop below nominal level when $h$ is large. The results are similar for other values of $s_1^*$. 

\begin{figure}
  \centering
  \includegraphics[width=0.7\textwidth]{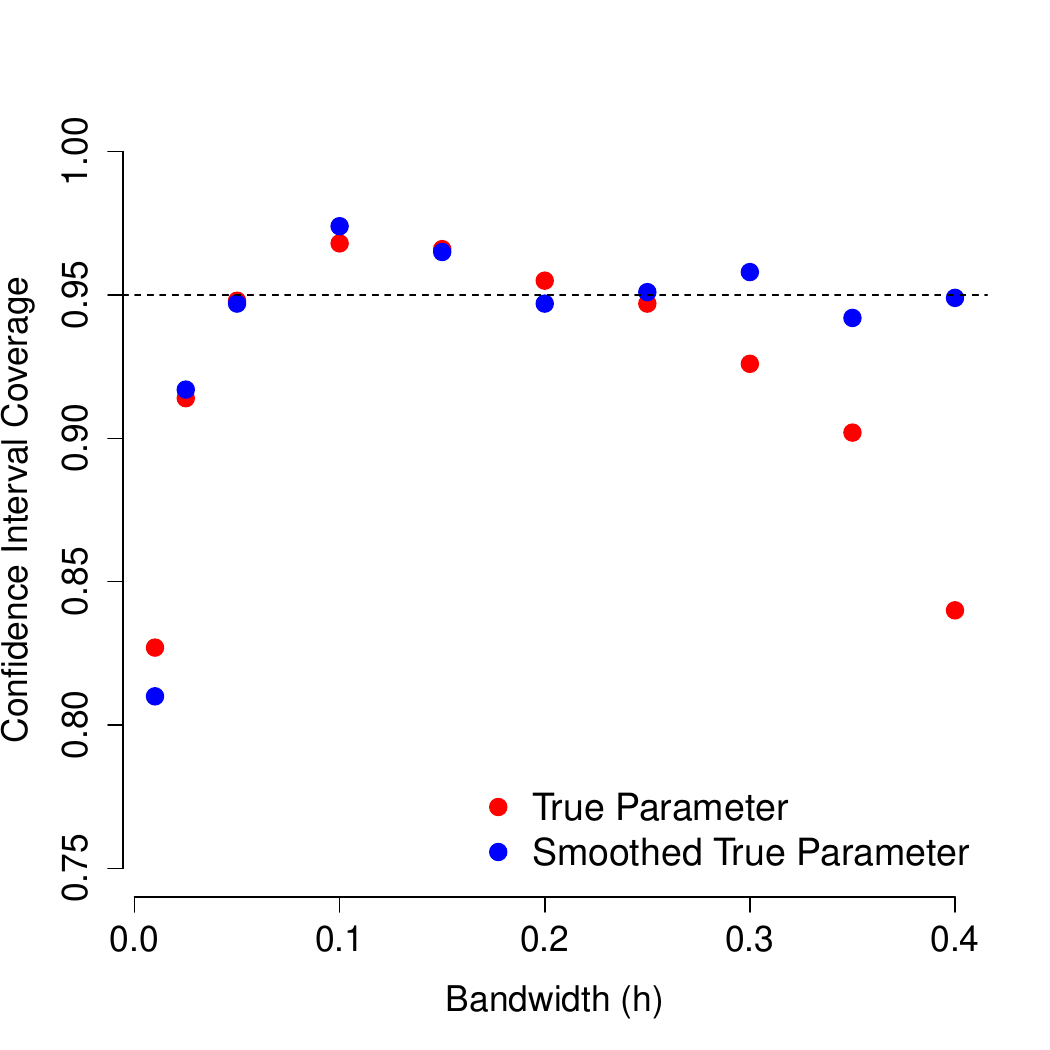}
  \caption{The coverage probability for the unsmoothed vaccine efficacy parameter drops below the nominal level when $h$ is small or large.  The coverage probability for the smoothed vaccine efficacy parameter drops below the nominal level when $h$ is small but does not drop below nominal level when $h$ is large.}
  \label{fig:cov}
\end{figure}

The primary results that we present in the main text will be based on a bandwidth of 0.2. We selected this bandwidth by first noting the three places that the bandwidth is used in our procedure: first, to estimate the density of $S|A=1,W$ that appears in the definition of $\Psi_1$; next, to estimate the density of $S|A=1,Y=1,W$ that appears in the definition of $\Psi_2$; and, finally, to estimate the density of $S^c|A=0,Y=0,W$ that appears in the definition of $\Psi_3$. To simplify selection of the bandwidth, we considered the estimation of the corresponding densities when the covariate $W$ was not conditioned on, namely the estimation of the densities of $(S|A=1)$, $(S|A=1,Y=1)$, and $(S^c|A=0,Y=0)$. Since $Y$ is a rare outcome, the sparsest stratum in these conditioning statements will occur in the second listed density, namely the density of $S|A=1,Y=1$. Consequently, we expect that the bandwidth selection should be driven by the estimation of $\Psi_2$. To select the bandwidth for our simulations, we generated many data sets of size 5000, estimated the density of $S|A=1,Y=1$, and evaluated the default bandwidth selected by the \texttt{density} function in \texttt{R}; this default bandwidth is selected based on Eq.~3.31 of \cite{Silverman86}. We found that values of approximately 0.15 were typically selected. Because the densities used in the definitions of $\Psi_1$, $\Psi_2$, and $\Psi_3$ are defined conditionally on $W$, we inflated this selected bandwidth to 0.2 to define the primary bandwidth considered in the main text. In Appendix~\ref{app:sim} we also present results at bandwidths of 0.1, 0.3, and 0.4.

Table \ref{tab:cov} shows the bias, standard error, and coverage probability of the unsmoothed and smoothed truth  for different $s_1^*$ with bandwidth chosen to be 0.2. The proposed TMLE estimator has low bias and is asymptotically unbiased. As expected, the coverage probability for the smoothed vaccine efficacy parameter is better than the coverage probability for the unsmoothed vaccine efficacy parameter. Our additional simulation results for bandwidths 0.1, 0.3, and 0.4 are given in Appendix~\ref{app:sim}. These results continue to show appropriate coverage of the smoothed parameter, though, as expected, when the bandwidth is large we see poorer coverage of the unsmoothed parameter.

\begin{table}[ht]
\centering
\begin{tabular}{lrrrrrrrrrrr}
  \hline
$s_1^*$ & 0 & 0.1 & 0.2 & 0.3 & 0.4 & 0.5 & 0.6 & 0.7 & 0.8 & 0.9  \\ 
  \hline
  Bias, Truth & -0.02 & -0.01 & -0.00 & 0.01 & 0.02 & 0.03 & 0.04 & 0.05 & 0.07 & 0.09 \\ 
  Bias, Smoothed & -0.01 & -0.02 & -0.02 & -0.01 & -0.01 & -0.01 & -0.01 & -0.00 & 0.01 & 0.03 \\ 
  Standard Error & 0.18 & 0.17 & 0.17 & 0.17 & 0.17 & 0.18 & 0.19 & 0.21 & 0.23 & 0.27 \\ 
  Coverage, Truth & 0.95 & 0.95 & 0.95 & 0.95 & 0.94 & 0.95 & 0.95 & 0.96 & 0.96 & 0.97 \\ 
  Coverage, Smoothed & 0.96 & 0.95 & 0.95 & 0.94 & 0.93 & 0.94 & 0.95 & 0.95 & 0.96 & 0.96 \\ 
   \hline
\end{tabular}
\caption{Bias, standard error, and coverage probability of the log relative risk point estimator and confidence intervals for different $s_1^*$ values with bandwidth chosen to be $0.2$.}
\label{tab:cov}
\end{table}

\section{Discussion}\label{sec:disc}
We have presented a targeted minimum loss-based estimator (TMLE) for estimating contrasts between the principally stratified means of absorbent endpoints in a crossover design, where we recall that an absorbent endpoint is an endpoints whose occurrence renders any subsequent measures of the biomarker of interest scientifically uninteresting. For most biomarkers, death is an absorbent endpoint. In HIV vaccine trials, HIV infection is an absorbent endpoint when the biomarker of interest is the immune response. We have established sufficient conditions for the nonparametric identifiability of treatment-specific principally stratified means. Our identifiability conditions do not require that the baseline covariates be discrete. We also established a necessary and sufficient condition for the falsifiability of our assumption that relates the conditional distribution of the biomarker in the population that was treated at baseline to the conditional distribution of the biomarker in the population that was treated at the crossover stage. An implication of this result is that the crossover assumption can be tested: if the test fails to reject, then there is no evidence in the data that the crossover assumption is falsifiable. If the test rejects, then there is evidence in the data that the crossover assumption is false.

Our proposed method does not handle right-censoring on the absorbent endpoint. When dropout is rare, this will at most induce a small amount of bias in an estimate of the principally stratified mean outcomes. When dropout is more common but rare enough that a reasonable proportion of subjects have follow-up to the end of the study, the methods in this work can be readily extended by incorporating inverse probability of censoring weights to the proposed estimators. When few subjects have follow-up until the end of the study, further methodological development is needed to develop low-variance estimators that account for dropout.

When the biomarker is continuous, our proposed method relies on a choice of bandwidth. Though our theoretical results provide guarantees for a wide range of sequences of selected bandwidths, we have given only limited discussion to how the bandwidth should be selected in practice. We did present one heuristic strategy in our Section~\ref{sec:sim}, which involves using existing bandwidth selection approaches for estimating the conditional density of $S$ given that $A=1$ and $Y=1$ -- this density is related to a density that appears in the definition of the parameter $\Psi_2$. We leave the theoretical analysis of this and other bandwidth selection strategies to future work.

Our proposed estimators are especially advantageous in settings where (i) the baseline covariates are predictive of the post-treatment biomarker and (ii) the relationship between these baseline covariates and the biomarker is complex. The predictiveness of the baseline covariates improves the plausibility of the identifiability assumption linking the biomarker measured in treated individuals to the biomarker measured in untreated individuals following their crossover to treatment. In these settings, obtaining statistical inference for the contrast of interest using existing methods relies on either ignoring the baseline covariates, discretizing the baseline covariates, or assuming a parametric relationship between the biomarker and these baseline covariates. Our proposed TMLE is especially compelling in these settings because it enables the incorporation of flexible estimation techniques for the conditional distribution of the biomarker given the baseline covariates, while still admitting statistical inference.

{\singlespacing\section*{Acknowledgements}
This work was partially supported by the National Institute of Allergy and
Infectious Disease at the National Institutes of Health under award number UM1 AI068635. The content is solely the responsibility of
the authors and does not necessarily represent the official views of the National
Institutes of Health. Alex Luedtke gratefully acknowledges the support of the New Development Fund
  of the Fred Hutchinson Cancer  Research Center. The authors thank Peter Gilbert for helpful comments that greatly improved the quality of the manuscript.}

{\singlespacing
\bibliographystyle{Chicago}
 \bibliography{persrule}

\begin{thebibliography}{40}
\providecommand{\natexlab}[1]{#1}
\providecommand{\url}[1]{\texttt{#1}}
\expandafter\ifx\csname urlstyle\endcsname\relax
  \providecommand{\doi}[1]{doi: #1}\else
  \providecommand{\doi}{doi: \begingroup \urlstyle{rm}\Url}\fi

\bibitem[Anthony and Bartlett(1999)]{Anthony&Bartlett1999}
M~Anthony and P~L Bartlett.
\newblock {Neural network learning: theoretical foundations}.
\newblock 1999.

\bibitem[Bareinboim and Pearl(2012)]{Bareinboim&Pearl2012}
E~Bareinboim and J~Pearl.
\newblock {Transportability of Causal Effects: Completeness Results}.
\newblock In \emph{Twenty-Sixth AAAI Conference on Artificial Intelligence},
  2012.

\bibitem[Bibaut and van~der Laan(2017)]{Bibaut&vanderLaan2017}
A~F Bibaut and M~J van~der Laan.
\newblock {Data-adaptive smoothing for optimal-rate estimation of possibly
  non-regular parameters}.
\newblock \emph{arXiv preprint arXiv:1706.07408}, 2017.

\bibitem[Bickel et~al.(1993)Bickel, Klaassen, Ritov, and Wellner]{Bickel1993}
P~J Bickel, C~A~J Klaassen, Y~Ritov, and J~A Wellner.
\newblock \emph{{Efficient and adaptive estimation for semiparametric models}}.
\newblock Johns Hopkins University Press, Baltimore, 1993.

\bibitem[Chiu(1991)]{Chiu1991}
S-T Chiu.
\newblock {Bandwidth selection for kernel density estimation}.
\newblock \emph{The Annals of Statistics}, pages 1883--1905, 1991.

\bibitem[Cole and Hern{\'{a}}n(2002)]{Cole&Hernan2002}
S~R Cole and M~A Hern{\'{a}}n.
\newblock {Fallibility in estimating direct effects}.
\newblock \emph{International journal of epidemiology}, 31\penalty0
  (1):\penalty0 163--165, 2002.

\bibitem[Corey et~al.(2015)Corey, Gilbert, Tomaras, Haynes, Pantaleo, and
  Fauci]{Corey2015}
Lawrence Corey, Peter~B Gilbert, Georgia~D Tomaras, Barton~F Haynes, Giuseppe
  Pantaleo, and Anthony~S Fauci.
\newblock Immune correlates of vaccine protection against hiv-1 acquisition.
\newblock \emph{Science translational medicine}, 7\penalty0 (310):\penalty0
  310rv7--310rv7, 2015.

\bibitem[Fan and Marron(1992)]{Fan&Marron1992}
J~Fan and J~S Marron.
\newblock {Best possible constant for bandwidth selection}.
\newblock \emph{The Annals of Statistics}, pages 2057--2070, 1992.

\bibitem[Follmann(2006)]{Follmann2006}
D~Follmann.
\newblock {Augmented designs to assess immune response in vaccine trials}.
\newblock \emph{Biometrics}, 62\penalty0 (4):\penalty0 1161--1169, 2006.

\bibitem[Frangakis and Rubin(2002)]{Frangakis&Rubin2002}
C~E Frangakis and D~B Rubin.
\newblock {Principal stratification in causal inference}.
\newblock \emph{Biometrics}, 58\penalty0 (1):\penalty0 21--29, 2002.

\bibitem[Gabriel and Follmann(2016)]{Gabriel&Follmann2016}
E~E Gabriel and D~Follmann.
\newblock {Augmented trial designs for evaluation of principal surrogates}.
\newblock \emph{Biostatistics}, 17\penalty0 (3):\penalty0 453--467, 2016.

\bibitem[Gabriel and Gilbert(2013)]{Gabriel&Gilbert2013}
E~E Gabriel and P~B Gilbert.
\newblock {Evaluating principal surrogate endpoints with time-to-event data
  accounting for time-varying treatment efficacy}.
\newblock \emph{Biostatistics}, 15\penalty0 (2):\penalty0 251--265, 2013.

\bibitem[Gilbert and Hudgens(2008)]{Gilbert&Hudgens2008}
P~B Gilbert and M~G Hudgens.
\newblock {Evaluating candidate principal surrogate endpoints}.
\newblock \emph{Biometrics}, 64\penalty0 (4):\penalty0 1146--1154, 2008.

\bibitem[Gilbert et~al.(2011)Gilbert, Hudgens, and Wolfson]{Gilbertetal2011}
P~B Gilbert, M~G Hudgens, and J~Wolfson.
\newblock {Commentary on" Principal Stratification--a Goal or a Tool?" by Judea
  Pearl}.
\newblock \emph{Int J Biostat}, 7\penalty0 (1):\penalty0 1--15, 2011.

\bibitem[Hall et~al.(1991)Hall, Sheather, Jones, and Marron]{Halletal1991}
P~Hall, S~J Sheather, M~C Jones, and J~S Marron.
\newblock {On optimal data-based bandwidth selection in kernel density
  estimation}.
\newblock \emph{Biometrika}, 78\penalty0 (2):\penalty0 263--269, 1991.

\bibitem[Hastie and Tibshirani(1990)]{Hastie&Tibshirani90}
T~J Hastie and R~J Tibshirani.
\newblock \emph{{Generalized additive models}}.
\newblock Chapman {\&} Hall, London, 1990.

\bibitem[Horowitz(2009)]{Horowitz2009}
J~L Horowitz.
\newblock \emph{{Semiparametric and nonparametric methods in econometrics}},
  volume~12.
\newblock Springer, 2009.

\bibitem[Huang et~al.(2013)Huang, Gilbert, and Wolfson]{Huangetal2013}
Y~Huang, P~B Gilbert, and J~Wolfson.
\newblock {Design and estimation for evaluating principal surrogate markers in
  vaccine trials}.
\newblock \emph{Biometrics}, 69\penalty0 (2):\penalty0 301--309, 2013.

\bibitem[Kennedy(2016)]{Kennedy2016}
E~H Kennedy.
\newblock {Semiparametric theory and empirical processes in causal inference}.
\newblock In \emph{Statistical Causal Inferences and Their Applications in
  Public Health Research}, pages 141--167. Springer, 2016.

\bibitem[Kennedy et~al.(2016)Kennedy, Ma, McHugh, and Small]{Kennedyetal2016}
E~H Kennedy, Z~Ma, M~D McHugh, and D~S Small.
\newblock {Non-parametric methods for doubly robust estimation of continuous
  treatment effects}.
\newblock \emph{J. Royal Stat. Soc Series B}, 2016.

\bibitem[Marsh(2016)]{Marsh2016}
T~Marsh.
\newblock {Efficient inference for an additive gene-treatment interaction from
  a nested two-phase study}.
\newblock 2016.

\bibitem[Nason and Follmann(2010)]{Nason&Follmann2010}
M~Nason and D~Follmann.
\newblock {Design and analysis of crossover trials for absorbing binary
  endpoints}.
\newblock \emph{Biometrics}, 66\penalty0 (3):\penalty0 958--965, 2010.

\bibitem[Pearl(2001)]{Pearl2001}
J~Pearl.
\newblock {Direct and indirect effects}.
\newblock In \emph{Proceedings of the Seventeenth Conference on Uncertainty in
  artificial intelligence}, pages 411--420, San Francisco, 2001. Morgan
  Kaufmann.

\bibitem[Pfanzagl(1990)]{Pfanzagl1990}
J~Pfanzagl.
\newblock \emph{{Estimation in semiparametric models}}.
\newblock Springer, 1990.

\bibitem[Polley and van~der Laan(2012)]{SuperLearnerPackage}
E~Polley and M~van~der Laan.
\newblock {Super Learner Prediction}, 2012.

\bibitem[Robins and Greenland(1992)]{RobinsGreenland1992}
J~M Robins and S~Greenland.
\newblock {Identifiability and exchangeability for direct and indirect
  effects}.
\newblock \emph{Epidemiol}, 3:\penalty0 143--155, 1992.

\bibitem[Rose and van~der Laan(2011)]{Rose2011}
S~Rose and M~J van~der Laan.
\newblock {A Targeted Maximum Likelihood Estimator for Two-Stage Designs}.
\newblock \emph{Int J Biostat}, 7\penalty0 (17), 2011.

\bibitem[Silverman(1986)]{Silverman86}
B~W Silverman.
\newblock \emph{{Density Estimation for Statistics and Data analysis}}.
\newblock Chapman {\&} Hall, 1986.

\bibitem[van~der Laan(2017)]{vanderLaan2017}
M~van~der Laan.
\newblock {A Generally Efficient Targeted Minimum Loss Based Estimator based on
  the Highly Adaptive Lasso}.
\newblock \emph{Int J Biostat}, 2017.

\bibitem[van~der Laan and Robins(2003)]{vdL02}
M~J van~der Laan and J~M Robins.
\newblock \emph{{Unified methods for censored longitudinal data and
  causality}}.
\newblock Springer, New York Berlin Heidelberg, 2003.

\bibitem[van~der Laan and Rose(2011)]{vanderLaan&Rose11}
M~J van~der Laan and S~Rose.
\newblock \emph{{Targeted Learning: Causal Inference for Observational and
  Experimental Data}}.
\newblock Springer, New York, New York, 2011.

\bibitem[van~der Laan et~al.(2013)van~der Laan, Hubbard, and
  Kherad]{vanderLaan&Hubbard13}
M~J van~der Laan, A~E Hubbard, and S~Kherad.
\newblock {Statistical inference for data adaptive target Parameters}.
\newblock Technical Report 314, Division of Biostatistics, University of
  California, Berkeley, 2013.

\bibitem[van~der Vaart and Wellner(1996)]{vanderVaartWellner1996}
A~W van~der Vaart and J~A Wellner.
\newblock \emph{{Weak convergence and empirical processes}}.
\newblock Springer, Berlin Heidelberg New York, 1996.

\bibitem[VanderWeele(2015)]{Vanderweele2015}
T~VanderWeele.
\newblock \emph{{Explanation in causal inference: methods for mediation and
  interaction}}.
\newblock Oxford University Press, 2015.

\bibitem[VanderWeele(2008)]{VanderWeele2008}
T~J VanderWeele.
\newblock {Simple relations between principal stratification and direct and
  indirect effects}.
\newblock \emph{Statistics {\&} Probability Letters}, 78\penalty0
  (17):\penalty0 2957--2962, 2008.

\bibitem[Wasserman(2006)]{Wasserman2006}
L~Wasserman.
\newblock \emph{{All of Nonparametric Statistics}}.
\newblock Springer-Verlag New York, Inc., Secaucus, NJ, USA, 2006.
\newblock ISBN 0387251456.

\bibitem[Wolfson and Gilbert(2010)]{Wolfson&Gilbert2010}
J~Wolfson and P~Gilbert.
\newblock {Statistical identifiability and the surrogate endpoint problem, with
  application to vaccine trials}.
\newblock \emph{Biometrics}, 66\penalty0 (4):\penalty0 1153--1161, 2010.

\bibitem[Wolfson and Henn(2014)]{Wolfson&Henn2014}
J~Wolfson and L~Henn.
\newblock {Hard, harder, hardest: principal stratification, statistical
  identifiability, and the inherent difficulty of finding surrogate endpoints}.
\newblock \emph{Emerging themes in epidemiology}, 11\penalty0 (1):\penalty0 14,
  2014.

\bibitem[Woods et~al.(1989)Woods, Williams, and Tavel]{Woodsetal1989}
J~R Woods, J~G Williams, and M~Tavel.
\newblock {The two-period crossover design in medical research}.
\newblock \emph{Ann Intern Med}, 110\penalty0 (7):\penalty0 560--566, 1989.

\bibitem[Zheng and van~der Laan(2011)]{Zheng&vanderLaan11}
Wenjing Zheng and Mark~J van~der Laan.
\newblock Cross-validated targeted minimum-loss-based estimation.
\newblock In \emph{Targeted Learning}, pages 459--474. Springer, 2011.

\end{thebibliography}
}
\appendix
\setcounter{equation}{0}
\renewcommand{\theequation}{A.\arabic{equation}}
\setcounter{theorem}{0}
\renewcommand{\thetheorem}{A.\arabic{theorem}}
\renewcommand{\thecorollary}{A.\arabic{theorem}}
\renewcommand{\thelemma}{A.\arabic{theorem}}
\renewcommand{\theproposition}{A.\arabic{theorem}}
\renewcommand{\theconjecture}{A.\arabic{theorem}}

\section*{Appendix}
\section{Proofs of Results from Main Text}\label{app:proofs}
\subsection{Identifiability}\label{app:ident}
\subsubsection{Identifiability under \ref{it:sutvacons}, \ref{it:ignorable}, and \ref{it:crossover}}
\begin{proof}[Proof of Theorem~\ref{thm:identifiability}]
Proof of \ref{it:margS1}. Using that $P^F(S_1=s_1^\star) = \E^F[P^F(S_1=s_1^\star|W)]$, ignorability shows that $P^F(S_1=s_1^\star|A=1,W)$, and consistency then shows that the right-hand side is equal to $\E[P(S=s_1^\star|A=1,W)]$.\\[0.5em]
Proof of \ref{it:Y1andS1}. Using that $P^F(Y_1=1,S_1=s_1^\star) = \E^F\left[P^F(Y_1=1,S_1=s_1^\star|W)\right]$, ignorability shows that $P^F(Y_1=1,S_1=s_1^\star|W)=P^F(Y_1=1,S_1=s_1^\star|A=1,W)$, and consistency shows that the right-hand side is equal to $\E\left[P(Y=1,S=s_1^\star|A=1,W)\right]$.\\[0.5em]
Proof of \ref{it:notY0andS1}. Note that
\begin{align*}
&P^F\left(Y_0=0,S_1=s_1^\star\right) = \E^F\left[\E^F\left[\Ind_{\left\{Y_0=0,S_1=s_1^\star\right\}}|W\right]\right] \tag{Tower Property} \\
&= \E^F\left[P^F\left(S_1=s_1^\star\middle|Y_0=0,W\right)P^F(Y_0=0|W)\right] \tag{Simplification} \\
&= \E^F\left[P^F\left(S_0^c=s_1^\star\middle|Y_0=0,W\right)P^F(Y_0=0|W)\right] \tag{\ref{it:crossover}} \\
&= \E^F\left[P^F\left(Y_0=0,S_0^c=s_1^\star\middle|W\right)\right] \tag{Simplification} \\
&= \E^F\left[P^F\left(Y_0=0,S_0^c=s_1^\star\middle|A=0,W\right)\right] \tag{\ref{it:ignorable}} \\
&= \E\left[P\left(Y=0,S^c=s_1^\star\middle|A=0,W\right)\right]. \tag{\ref{it:sutvacons}}
\end{align*}
\end{proof}

\subsubsection{Non-Falsifiability of \ref{it:crossover}}\label{app:falsifiability}
We now show that \ref{it:crossover} is not falsifiable from the observed data if $\Psi_4(P)=0$. We do this by constructing a distribution that is consistent with the observed data (in the sense of \ref{it:sutvacons}), satisfies the ignorability assumption \ref{it:ignorable}, and satisfies the crossover assumption \ref{it:crossover}. This shows that, if $\Psi_4(P)=0$, then we cannot falsify \ref{it:crossover} from the observed data, even as the sample size approaches infinity.
\begin{theorem}\label{thm:falsifiability}
If $\Psi_4(P)=0$, then condition \ref{it:crossover} is not falsifiable from the observed data.
\end{theorem}

\begin{proof}
Suppose that $\Psi_4(P)=0$. We establish that \ref{it:crossover} is not falsifiable by showing that there exists a distribution $P_1^F$ such that, if $P^F=P_1^F$, then \ref{it:sutvacons}, \ref{it:ignorable}, and \ref{it:crossover} hold. Consider the distribution $P_1^F$, defined by its conditional densities (with respect to appropriate dominating measures):
\begin{align*}
& p_1^F(y_1 | s_0^c, y_0, s_1, a, w) = p(y|A=1,s_1,w), \\
&p_1^F(s_0^c , | s_1 , y_0 , a, w) = p(S^c = s_0^c|Y=0,A=0,w) \\
& p_1^F(s_1, y_0 | a, w) = \begin{cases}
p(S^c=s_1,Y=0|A=0,w),&\mbox{ if } y_0=0 \\
p(S=s_1|A=1,w)-p(S^c=s_1,Y=0|A=0,w),&\mbox{ if } y_0=1, \\
\end{cases} \\
&p_1^F(a,w) = p(a,w).
\end{align*}
Given that $P$ is a probability measure, clearly, the first, second, and fourth conditional densities are everywhere nonnegative and sum to one. The third conditional density is nonnegative because $\Psi_4(P)=0$ and sums to one since
\begin{align*}
\sum_{s_1} \sum_{y_0} p_1^F(s_1, y_0 | a, w) = \sum_{s_1} p(S=s_1|A=1,w) = 1.
\end{align*}
Condition \ref{it:ignorable} is easily established for the above since $a$ does not appear on the right-hand side of any of the density definitions except for the definition of $p_1^F(a,w)$. To establish \ref{it:sutvacons}, the only challenge is in showing that $p_1^F(s_1 | w) = p(S=s_1 | A=1, w)$ and $p_1^F(y_0 | w) = p(Y=y_0 | A=0, w)$, but these results both follow by respectively summing $p_1^F(s_1, y_0 | a, w)$ over $y_0$ and $s_1$. Furthermore, note that \ref{it:crossover} holds if $P^F=P_1^F$, since, for all $s_1$, the fact that $S_1$ is independent of $A,S_0^c$ given $W$ shows that $p_1^F(S_1=s_1 | Y_0=0, w)$ equals $p(S^c=s_1|Y=0,A=0,W)$, and \ref{it:sutvacons} and \ref{it:ignorable} show that $p(S^c=s_1|Y=0,A=0,W) = p_1^F(S_0^c=s_1|Y_0=0,w)$.
\end{proof}

\subsection{Estimation when Biomarker is Discrete}\label{app:est}
\begin{proof}[Proof of Theorem~\ref{thm:eif}]
It is well known \citep{vdL02} that, for an arbitrary bounded function $f : \mathcal{O}\rightarrow\mathbb{R}$, $a\in\{0,1\}$, and parameter $\Phi^{f,a}$ defined by $\Phi^{f,a}(P')\equiv \E_{P'}[\E_{P'}[ f(O)|A=a,W ]]$, the gradient at $P'$ is
\begin{align*}
o\mapsto \frac{\indicator{A=a}}{P'(a|w)}\left\{f(o) - \E_{P'}[f(o)|A,W] \right\} + \E_{P'}[f(o)|A=a,W] -  \Phi^{f,a}(P').
\end{align*}
The proof concludes by noting that $\Psi_k(P')=\Phi^{f_k,a_k}(P')$ for all $k=1,2,3$.
\end{proof}

\section{Estimation when Principal Biomarker is Continuous}\label{app:cont}
We first present several lemmas used to prove Theorem~\ref{thm:estsmooth} from the main text, which establishes an asymptotically normal distribution for the estimate of the smoothed parameter. The proof of these lemmas is deferred until after that of Theorem~\ref{thm:estsmooth}. We then prove Theorem~\ref{thm:bias} from the main text, which shows that, under some conditions, the smoothed parameter converges to the true, unsmoothed parameter as the bandwidth shrinks to zero.

We now present the lemmas used in the above proof of Theorem~\ref{thm:estsmooth}. The first lemma uses fixed functions $w\mapsto \check{q}_k(w)$ and $w\mapsto \check{\pi}_k(w)$. When we invoke this theorem, we will do so at $\hat{q}_{k,h}^{v}$ and $\hat{P}(a_k|W=\cdot)$, where we will do to this conditionally on training sample $v$ so that we can treat these functions as known.
\begin{lemma}\label{lem:vc}
Fix $k$ and functions $w\mapsto \check{q}_k(w)$ and $w\mapsto \check{\pi}_k(w)$. The class $\mathcal{G}_k(\check{q}_k,\check{\pi}_k)$ defined in \eqref{eq:Gdef} is VC-subgraph \citep{vanderVaartWellner1996}. Furthermore, the VC dimension of this class does not depend on $\check{q}_k,\check{\pi}_k$.
\end{lemma}
The next lemma is useful for establishing that our estimator ensures that the empirical mean of the smoothed influence function is zero.
\begin{lemma}\label{lem:zest}
Fix $k\in\{1,2,3\}$. Define
\begin{align*}
Z_{k,\hat{h}_n}(\epsilon)&\equiv \frac{1}{n}\sum_{i=1}^n \frac{\Ind_{A_i=a_k}}{\hat{P}(A_i|W_i)} \left[f_{k,\hat{h}_n}(O_i) - \exp\left\{\log \hat{q}_{k,\hat{h}_n}^{v(i)}(W_i) + \epsilon\right\}\right].
\end{align*}
If \ref{it:estsbounded}, then $Z_{k,\hat{h}_n}(\hat{\epsilon}_k) = 0$ with probability approaching one.
\end{lemma}
The final lemma presents a technical condition controlling a class used in the proof of Theorem~\ref{thm:estsmooth}.
\begin{lemma}\label{lem:envelope}
Fix $\delta>0$. Under the conditions of Theorem~\ref{thm:estsmooth}, the class $\widetilde{\mathcal{G}}_{k,\delta,\bar{h}_n}^v$ defined in the proof of Theorem~\ref{thm:estsmooth} has envelope $c\delta$ for a constant $c$ that does not depend on the sample realization $O_1,\ldots,O_n$.
\end{lemma}

\begin{proof}[Proof of Theorem~\ref{thm:estsmooth}]
To simplify notation, we assume that $n$ is divisible by 10 so that validation fold $v$ is of size $n/10$ -- note that the results of this proof go through without this assumption. We let $P_n^v$ denote the empirical distribution of the observations in validation set $v$. For a function $g$ and a distribution $Q$, we let $Qg = \E_Q[g(O)]$.

Fix $k\in\{1,2,3\}$. By the same expansion used to establish \eqref{eq:rem},
\begin{align}
\Psi_{k,h}(P') - \Psi_{k,h}(P)&= -P D_{k,h}^{P'} + \Rem_{k,h}(P,P'), \label{eq:ident}
\end{align}
where, with $q'_{k,h}(w)\equiv \E_{P'}[f_{k,h}(O)|a_k,w]$,
\begin{align*}
\Rem_{k,h}(P,P')&= \E_P\left[\frac{\left(P'(a_k|W)-P(a_k|W)\right)\left(q'_{k,h}(W)-q_{k,h}(W)\right)}{P'(a_k|W)}\right].
\end{align*}
By Cauchy-Schwarz and \ref{it:estsbounded}, this remainder upper bounds as
\begin{align*}
\left|\Rem_{k,h}(P,P')\right|\le \delta^{-1}&\norm{P'(a_k|W=\cdot)-P(a_k|W=\cdot)}\norm{q'_{k,h}-q_{k,h}}.
\end{align*}
For a given training fold $v$, the above is $o_P(n^{-1/2}\hat{h}_n^{-1/2})$ by \ref{it:rem} when evaluated at $h=\hat{h}_n$ and $P'=\hat{P}_{k,\hat{h}_n}^{v,*}$. Multiplying by $n^{-1}$ and summing the identity \eqref{eq:ident} over $v=1,\ldots,10$ and $i\in v$ at $h=\hat{h}_n$ and $P'=\hat{P}_{k,\hat{h}_n}^{v,*}$, shows that
\begin{align*}
\hat{\psi}_{k,h} - \Psi_{k,h}(P)&= -\frac{1}{n}\sum_{v=1}^{10} \sum_{i: v(i)=1} \E\left[D_{k,\hat{h}_n}^{\hat{P}_{k,\hat{h}_n}^{v,*}}(O)\right] + o_P(n^{-1/2}\hat{h}_n^{-1/2}).
\end{align*}
Using Lemma~\ref{lem:zest}, with probability approaching $1$ the above shows that
\begin{align}
\hat{\psi}_{k,h} - \Psi_{k,h}(P)&= \frac{1}{10}\sum_{v=1}^{10} (P_n^v-P)D_{k,\hat{h}_n}^{\hat{P}_{k,\hat{h}_n}^{v,*}} + o_P(n^{-1/2}\hat{h}_n^{-1/2}). \label{eq:preVCapplication}
\end{align}
Consider the following $\mathbb{R}^6$-valued data structure, which can be derived from the observed data structure: $\underbar{O}\equiv (\check{q}_k(W),\check{\pi}_k(W),A,S,S^c,Y)$. For $i=1,\ldots,n$, we will let $\underbar{O}_i$ denote their corresponding $\mathbb{R}^6$-valued data structure. Noting that $f_{k,h}(o)$ can be evaluated from $\underbar{o}=(\check{q}_k(w),\check{\pi}_k(w),a,s,s^c,y)$, we define $\tilde{f}_{k,h}$ so that $\tilde{f}_{k,h}(\underbar{o})=f_{k,h}(o)$. For fixed functions $w\mapsto \check{q}_k(w)$ and $w\mapsto \check{\pi}_k$, define the class
\begin{align}
&\mathcal{G}_k(\check{q}_k,\check{\pi}_k)\equiv \left\{\underbar{o}\mapsto g_{k,\epsilon,h}^v(\underbar{o}) : \epsilon\in\mathbb{R},h\in\mathbb{R}\right\}, \label{eq:Gdef}
\end{align}
where
\begin{align}
g_{k,\epsilon,h}^v(\underbar{o})\equiv \frac{\Ind_{a=a_k}}{\check{\pi}_k(w)} \left[\tilde{f}_{k,h}(\underbar{o}) - \exp\left\{\log \check{q}_k(w) + \epsilon\right\}\right] + \exp\left[\log \check{q}_k(w) + \epsilon\right]. \label{eq:gepsh}
\end{align}
Lemma~\ref{lem:vc} demonstrates that $\mathcal{G}_k(\check{q}_k,\check{\pi}_k)$ is VC-subgraph \citep{vanderVaartWellner1996}, and that the VC dimension does not depend on $\check{q}_k,\check{\pi}_k$. For each $v$ and for the class $\mathcal{G}_k(\hat{q}_{k,h}^v,w\mapsto \hat{P}^v(a_k|w))$, let $g_{k,\epsilon,h}^v$ denote the element of this class indexed by $\epsilon$ and $h$. Note that \eqref{eq:preVCapplication} rewrites as
\begin{align}
\hat{\psi}_{k,h} - \Psi_{k,h}(P)&= \frac{1}{10}\sum_{v=1}^{10} (P_n^v - P) g_{k,\hat{\epsilon}_k,\hat{h}_n}^v + o_P(n^{-1/2}\hat{h}_n^{-1/2}). \label{eq:expansion}
\end{align}
Consider the following subclass of $\mathcal{G}_k(\hat{q}_{k,h}^v,w\mapsto \hat{P}^v(a_k|w))$:
\begin{align*}
\mathcal{G}_{k,\delta}^v\equiv \left\{\underbar{o}\mapsto g_{k,\epsilon,h}^v(\underbar{o}) : |\epsilon-\bar{\epsilon}_k|\le \delta,|1-\hat{h}_n/\bar{h}_n|\le \delta\bar{h}_n^{1/2}\right\}.
\end{align*}
The VC dimension of this subclass is no greater than that of $\mathcal{G}_k(\hat{q}_{k,h}^v,w\mapsto \hat{P}^v(a_k|w))$ and so is bounded by a constant that does not depend on the sample. Combining this with the fact that the VC dimension of this subclass is bounded by a universal constant $c_1$ shows that the VC dimension of the class $\widetilde{\mathcal{G}}_{k,\delta,\bar{h}_n}^v\equiv \left\{\bar{h}_n^{1/2}[g_\delta^v - g_{k,\bar{\epsilon}_k,\bar{h}_n}^v] : g_\delta^v\in \mathcal{G}_{k,\delta}^v\right\}$ is also bounded by a constant that does not rely on the sample. Introducing $\widetilde{\mathcal{G}}_{k,\delta,\bar{h}_n}^v$ is useful because, for any fixed $\delta>0$, \eqref{eq:expansion} combined with \ref{it:band} and \ref{it:barepsilon} yield that, with probability approaching one,
\begin{align}
&\left|\hat{\psi}_{k,h} - \Psi_{k,h}(P)-\frac{1}{10}\sum_{v=1}^{10} (P_n^v - P) g_{k,\bar{\epsilon}_k,\bar{h}_n}^v -o_P(n^{-1/2}\hat{h}_n^{-1/2})\right| \nonumber \\
&\le \left|\frac{1}{10}\sum_{v=1}^{10} (P_n^v - P) (g_{k,\hat{\epsilon}_k,\hat{h}_n}^v-g_{k,\bar{\epsilon}_k,\bar{h}_n}^v)\right|\le \frac{1}{10\bar{h}_n^{1/2}}\sum_{v=1}^{10} \sup_{\tilde{g}_\delta^v\in \widetilde{\mathcal{G}}_{k,\delta,\bar{h}_n}^v}\left|(P_n^v - P) \tilde{g}_\delta^v\right|. \label{eq:empprocbd}
\end{align}
Lemma~\ref{lem:envelope} shows that the class $\widetilde{\mathcal{G}}_{k,\delta,\bar{h}_n}^v$ has an envelope upper bounded by $c\delta$ for a constant $c$ that does not depend on the sample, i.e. $\sup_{g\in \widetilde{\mathcal{G}}_{k,\delta,\bar{h}_n}^v}|g(\underbar{o})|\le c\delta$ for all $\underbar{o}$. Now, for each $v$ we know that $\widetilde{\mathcal{G}}_{k,\delta,\bar{h}_n}^v$ is a VC-class and there exists an upper bound on the VC-dimension that does not depend on $O_1,\ldots,O_n$, and so Theorem~2.14.1 in \cite{vanderVaartWellner1996} can be used to show that, conditionally on the training set $v$,
\begin{align*}
n^{1/2}\E_{P_0^v}\left[\sup_{\tilde{g}_\delta^v\in \widetilde{\mathcal{G}}_{k,\delta,\bar{h}_n}^v}\left|(P_n^v - P) \tilde{g}_\delta^v\right|^2\right]^{1/2}&\lesssim \delta,
\end{align*}
where $P_0^v$ is the distribution of the data in validation fold $v$ and here and throughout we use ``$\lesssim$'' to denote ``less than or equal to up to a positive multiplicative constant that does not depend on the observations $O_1,\ldots,O_n$''. The above holds for all $\delta$. Therefore, for any deterministic sequence $\delta_n\downarrow 0$, the above shows that the right-hand side of \eqref{eq:empprocbd} is $o_P(n^{-1/2}\bar{h}_n^{-1/2})$ for this choice of $\delta=\delta_n$. Moreover, by \ref{it:band}, $o_P(n^{-1/2}\bar{h}_n^{-1/2})=o_P(n^{-1/2}\hat{h}_n^{-1/2})$.

We will choose the sequence $\delta_n\downarrow 0$ in such a way that the display in \eqref{eq:empprocbd} at $\delta=\delta_n$ holds with probability approaching one. We show that such a sequence exists as follows. Let $\mathcal{E}_{\delta,n}$ denote the event that \eqref{eq:empprocbd} does not hold. Because, for any fixed $\delta$, the complement of $\mathcal{E}_{\delta,n}$ holds with probability approaching one, we know that there exists a sequence $\varepsilon_{\delta,n}\downarrow 0$ so that $P\{\mathcal{E}_{\delta,n}\}\le \varepsilon_{\delta,n}$ for all $n$. As this is true for all $\delta$, there exists a sequence $\delta_n\downarrow 0$ so that $P\{\mathcal{E}_{\delta_n,n}\}\rightarrow 0$ as $n\rightarrow\infty$. One way to construct such a sequence $\{\delta_n\}$ is as follows: first, let $n_1=1$; then recursively, for $j=1,2,\ldots$, choose $\delta_n=1/2^j$ for all $n=n_j+1,\ldots,n_{j+1}$, where $n_{j+1}$ is the smallest natural number that is greater than $n_j$ such that $\max\{\varepsilon_{1/2^j,n_{j+1}},\varepsilon_{1/2^{j+1},n_{j+1}}\} \le 1/2^j$ --- notably, the fact that $\max\{\varepsilon_{1/2^j,n},\varepsilon_{1/2^{j+1},n}\}\rightarrow 0$ as $n\rightarrow\infty$ implies that $n_{j+1}$ is finite. The event $\mathcal{E}_{\delta_n,n}$ occurs with probability approaching zero along this deterministic sequence, and so the display in \eqref{eq:empprocbd} holds with probability approaching one when $\delta=\delta_n$. Plugging this choice of $\{\delta_n\}$ sequence into the discussion from the previous paragraph, which controlled the right-hand side of \eqref{eq:empprocbd} for a given sequence $\delta_n\downarrow 0$, we see that
\begin{align*}
\hat{\psi}_{k,h} - \Psi_{k,h}(P)= \frac{1}{10}\sum_{v=1}^{10} (P_n^v - P) g_{k,\bar{\epsilon}_k,\bar{h}_n}^v +o_P(n^{-1/2}\hat{h}_n^{-1/2}).
\end{align*}
For each $v$, the function $g_{k,\bar{\epsilon}_k,\bar{h}_n}^v$ does not depend on validation sample $v$. Therefore, Chebyshev's inequality applied over each $(P_n^v - P) g_{k,\bar{\epsilon}_k,\bar{h}_n}^v$ can be used to show that, if there exists a function $g_{k,\bar{\epsilon}_k,\bar{h}_n}$ such that $\bar{h}_n P_0^v[(g_{k,\bar{\epsilon}_k,\bar{h}_n}^v-g_{k,\bar{\epsilon}_k,\bar{h}_n})^2]=o_P(1)$, then $(P_n^v - P) g_{k,\bar{\epsilon}_k,\bar{h}_n}^v$ is equal to $(P_n^v - P) g_{k,\bar{\epsilon}_k,\bar{h}_n} + o_P(n^{-1/2}\bar{h}_n^{-1/2})$, where under \ref{it:band} the $o_P$ term can also be expressed as $o_P(n^{-1/2}\hat{h}_n^{-1/2})$. Under \ref{it:L2lim}, the convergence to a fixed function $g_{k,\bar{\epsilon}_k,\bar{h}_n}$ does indeed hold, with $g_{k,\bar{\epsilon}_k,\bar{h}_n}$ defined in \eqref{eq:glimdef}.
Hence,
\begin{align*}
\hat{\psi}_{k,h} - \Psi_{k,h}(P)&= \frac{1}{10}\sum_{v=1}^{10} (P_n^v - P) g_{k,\bar{\epsilon}_k,\bar{h}_n} + o_P(n^{-1/2}\hat{h}_n^{-1/2}) \\
&= (P_n-P) g_{k,\bar{\epsilon}_k,\bar{h}_n} + o_P(n^{-1/2}\hat{h}_n^{-1/2}).
\end{align*}
The above asymptotically linear expansion holds for each $k$.

Let $g_{\bar{\epsilon},\bar{h}_n} : o\mapsto (g_{k,\bar{\epsilon}_k,\bar{h}_n}(o) : k=1,2,3)$. 
Using \ref{it:L2lim} and \ref{it:estsbounded}, each component in the $\mathbb{R}^3$-valued vector $g_{\bar{\epsilon},\bar{h}_n}(O)$ is almost surely bounded by a universal constant times $\bar{h}_n^{-1}$ for $O\sim P$. Using that $\bar{h}_n$ and $\hat{h}_n$ are essentially equivalent under \ref{it:band}, our result will follow if we establish asymptotic normality of the sequence $n^{1/2}(P_n-P)\bar{h}_n^{1/2}\Sigma_n^{-1/2} g_{\bar{\epsilon},\bar{h}_n}$. We do this via a Cram\'{e}r-Wold device and a central limit theorem for triangular arrays. Let $b$ be an $\mathbb{R}^3$-valued column vector. Observe that \ref{it:covmat} ensures that, with probability approaching one as $n\rightarrow\infty$, $\bar{h}_n^{1/2} b^T \Sigma_n^{-1/2} g_{\bar{\epsilon},\bar{h}_n}(O)$ is almost surely bounded by a constant $c_b$ times $\bar{h}_n^{-1/2}$. Furthermore, $\Var_P[ \bar{h}_n^{1/2} \Sigma_n^{-1/2} b^T g_{\bar{\epsilon},\bar{h}_n}(O)]$ is equal to $\norm{b}^2$. Using that $n\bar{h}_n\rightarrow\infty$ (\ref{it:band}), one can show that the Lindeberg condition for triangular arrays is satisfied and the sequence $\langle b, n^{1/2}(P_n-P)\bar{h}_n^{1/2}\Sigma_n^{-1/2} g_{\bar{\epsilon},\bar{h}_n} \rangle$ converges in distribution to a normal random variable with variance $\norm{b}^2$. As $b$ was arbitrary, $n^{1/2}(P_n-P)\bar{h}_n^{1/2}\Sigma_n^{-1/2} g_{\bar{\epsilon},\bar{h}_n}$ converges to a $N(\mb{0},\Id)$ random variable.
\end{proof}

\begin{proof}[Proof of Lemma~\ref{lem:vc}]
The subgraph of a class $\mathcal{G}'\equiv\{\underbar{o}\mapsto g(\underbar{o}) : g\}$ is defined as the set $\{(\underbar{o},z) : z<g(\underbar{o})\}$. For a pair $(\underbar{o},z)$, membership to this set can be computed by first computing $g(\underbar{o})$, and subsequently returning an indicator that $z<g(\underbar{o})$. Now, for $\mathcal{G}_k$, these functions $g$ are indexed only by parameters $(\epsilon,h)\in\mathbb{R}^2$. Furthermore, for properly defined $t$ that only depends on the choice of kernel $K$, \ref{it:kern} ensures that all of these these functions can be computed using no more than a total of $t$ arithmetic operations ($+$, $-$, $\div$, $\times$), indicator functions ($x\mapsto \Ind\{x>c\}$ for a constant $c$), and exponential functions ($x\mapsto e^x$). By Theorem~8.14 in \cite{Anthony&Bartlett1999}, the VC-dimension of this class is no more than $m(t+2)(t+2+19\log_2[9(t+2)])$, hence is finite and does not depend on $\check{q}_k,\check{\pi}_k$.
\end{proof}

\begin{proof}[Proof of Lemma~\ref{lem:zest}]
Fix $k\in\{1,2,3\}$. By the strong positivity assumption, with probability approaching one there exists at least one $i$ such that $A_i=a_k$. Without loss of generality, suppose that this holds. Combining this with the bound on $\hat{q}_{k,\hat{h}_n}^v$ in \ref{it:estsbounded}, $Z_{k,\hat{h}_n}$ is a monotonically strictly decreasing. Furthermore, $Z_{k,\hat{h}_n}(\epsilon)$ respectively diverges to $+\infty$ and $-\infty$ as $\epsilon$ tends to $-\infty$ and $+\infty$. Furthermore, $Z_{k,\hat{h}_n}$ is continuous, and therefore there exists a (unique) solution in $\epsilon$ to $Z_{k,\hat{h}_n}(\epsilon)=0$, and this solution must coincide with $\hat{\epsilon}_k$, defined as the minimizer of $Z_{k,\hat{h}_n}(\epsilon)^2$ in $\epsilon$.
\end{proof}

\begin{proof}[Proof of Lemma~\ref{lem:envelope}]
Note that
\begin{align*}
&\bar{h}_n^{1/2} |g_{k,\epsilon,h}^v(\underbar{o})-g_{k,\bar{\epsilon}_k,\bar{h}_n}^v(\underbar{o})| \\
&\le \bar{h}_n^{1/2} \frac{\Ind_{a=a_k}}{\hat{P}_k(a_k|w)} \left|f_{k,h}(o)-f_{k,\bar{h}_n}(o)\right| + \bar{h}_n^{1/2}\left|\frac{\Ind_{a=a_k}}{\hat{P}_k(a_k|w)} -1\right|\left|e^{\log \hat{q}_{k,h}^v(w) + \epsilon}-e^{\log \hat{q}_{k,h}^v(w) + \bar{\epsilon}_k}\right| \\
&\le \bar{h}_n^{1/2}\hat{P}_k(a_k|w)^{-1}\left[\left|f_{k,h}(o)-f_{k,\bar{h}_n}(o)\right| + \left|e^{\log \hat{q}_{k,h}^v(w) + \epsilon}-e^{\log \hat{q}_{k,h}^v(w) + \bar{\epsilon}_k}\right|\right]
\end{align*}
By \eqref{it:kern}, $K$ is bounded and so $\left|f_{k,h}(o)-f_{k,\bar{h}_n}(o)\right|\lesssim |h^{-1}-\bar{h}_n^{-1}|$. Using a Taylor expansion and the bound on $\hat{q}_{k,h}^v$ from \eqref{it:estsbounded}, we also have that, with probability approaching one, $\left|e^{\log \hat{q}_{k,h}^v(w) + \epsilon}-e^{\log \hat{q}_{k,h}^v(w) + \bar{\epsilon}_k}\right|\lesssim |\epsilon-\bar{\epsilon}_k|$. Furthermore, by \eqref{it:estsbounded}, $\hat{P}_k(a_k|w)^{-1}\lesssim 1$. Hence,
\begin{align*}
|g_{k,\epsilon,h}^v(\underbar{o})-g_{k,\bar{\epsilon}_k,\bar{h}_n}^v(\underbar{o})|&\lesssim \bar{h}_n^{1/2}|h^{-1}-\bar{h}_n^{-1}| - \bar{h}_n^{1/2}|\epsilon-\bar{\epsilon}_k| \\
&= \bar{h}_n^{-1/2}\left|1-\frac{h}{\bar{h}_n}\right|\left(1 + \left[\frac{\bar{h}_n}{h}-1\right]\right)- \bar{h}_n^{1/2}|\epsilon-\bar{\epsilon}_k|
\end{align*}
If $g_{k,\epsilon,h}^v$ falls in $\mathcal{G}_{k,\delta}^v$, then the right-hand side upper bounds by $c \delta$ for a constant $c$ that does not depend on $O_1,\ldots,O_n$. Therefore, the envelope of the class $\widetilde{\mathcal{G}}_{k,\delta,\bar{h}_n}^v$ is of the order $\delta$.
\end{proof}

\begin{proof}[Proof of Theorem~\ref{thm:bias}]
Fix $k\in\{1,2,3\}$ and suppose that $r\ge t+1$. A Taylor series expansion shows that
\begin{align*}
\Psi_{k,h}(P) - \Psi_{k}(P)=\,& \int \left[\Psi_k(P;s) - \Psi_{k}(P)\right] K_h(s-s_1^\star) ds \\
=\,& \sum_{j=1}^t \left[\frac{1}{j!}\left.\frac{\partial^j}{\partial s^j} \Psi_{k}(P;s)\right|_{s=s_1^\star} \int (s-s_1^\star)^j K_h(s-s_1^\star) ds \right] \\
&+ \int \frac{(s-s_1^\star)^{t+1}}{(t+1)!} \left.\frac{\partial^{t+1}}{\partial s^{t+1}} \Psi_k(P;s) \right|_{s=\tilde{s}_s} K_h(s-s_1^\star) ds,
\end{align*}
where each $\tilde{s}_s$ is an intermediate value falling between $s$ and $s_1^\star$. So, the right-hand side is equal to the final term. By the uniform bound on the $(t+1)^{\textnormal{th}}$ derivative of $s\mapsto \Psi_k(P;s)$, the magnitude of the final term is upper bounded by a $t$-dependent constant times
\begin{align*}
\int \left|(s-s_1^\star)^{t+1} K_h(s-s_1^\star)\right| ds&= h^{t+1} \int \left|\left(\frac{s-s_1^\star}{h}\right)^{t+1} K\left(\frac{s-s_1^\star}{h}\right)\right| \frac{ds}{h} = h^{t+1} \int \left|u^{t+1} K(u)\right| du.
\end{align*}
The right-hand side is $O(h^{t+1})$ by assumption. The result follows because $k\in\{1,2,3\}$ was arbitrary.
\end{proof}

\section{Efficient Influence Functions}\label{app:EIF}
We now review the definition of efficient influence functions, which is used in Theorem~\ref{thm:eif}. Define the following fluctuation submodel through $P$:
\begin{align*}
dP_{\epsilon}(o)&\equiv \left[1 + \epsilon h(o)\right]dP(o), \mbox{ where }\E_P[h(O)] = 0\textrm{ and }\sup_o|h(o)|<\infty.
\end{align*}
The function $h$ is a score, and the closure of the linear span of all scores yields the tangent space. It is the resulting tangent space that is important, as pathwise differentiability is equivalent for any set of functions $h$ that yield the same tangent space. Hence, the restriction that $\sup_{o\in\mathcal{O}}|h(o)|<\infty$, while convenient, will have no impact on the resulting differentiability properties.

The parameter $\Psi$ is called pathwise differentiable at $P$ if there exists a $D^P\in L_0^2(P)$, i.e. a $D^P$ such that $D^P(O)$ has mean zero and finite variance under $O\sim P$, such that
\begin{align*}
\Psi(P_\epsilon)-\Psi(P)&= \epsilon \int D^P(o) h(o) dP(o) + o(\epsilon).
\end{align*}
We call $D^P$ a gradient of $\Psi$ at $P$. The efficient influence function is the gradient $D^P$ for which $D^P(O)$, $O\sim P$, has minimal variance. In a nonparametric model, $D^P(O)$ is almost surely unique.

\newpage
\section{Two-Phase Sampling Algorithms}\label{app:twophase}

\begin{algorithm}
\caption{TMLE for Estimating $\Psi_1(P)$, $\Psi_2(P)$, and $\Psi_3(P)$ under Two-Phase Sampling}\label{alg:2phasetmle}
\begin{algorithmic}
\Statex Takes as input $n$ observations $\textnormal{Obs}\equiv\{O_i : i=1,\ldots,n\}$, and also the weights $\Delta_i\pi_i^{-1}$.
\Function{TMLE}{$\textnormal{Obs}$}
	\State \parbox[t]{\dimexpr\linewidth-\algorithmicindent}{\textbf{Stabilization:} For $a=0,1$, let $c(a)\equiv \frac{\sum_{i : A_i=a}  \Delta_i/\pi_i}{\#\{i : A_i=a\}}$. For each $i$, let $\bar{\pi}_i\equiv c(A_i)\pi_i$.\\
	\Comment{For $a=0,1$, the choice of $c(a)$ ensures that $\frac{1}{\#\{i : A_i=a\}}\sum_{i : A_i=a}\Delta_i\bar{\pi}_i^{-1} = 1$.}}
	\State \parbox[t]{\dimexpr\linewidth-\algorithmicindent}{\textbf{Initial Estimates:} Define an initial estimator $\hat{P}$ of $P$:
\begin{itemize}
	\item The marginal distribution of $W$ under $\hat{P}$ should be the weighted empirical distribution, for which each observations puts a mass of $\Delta_i\bar{\pi}_i^{-1}$ at each observation $\tilde{O}_i$.
	\item Only the estimates of the components of $P(O|W)$ needed to evaluate $\Psi_k$ and $D_k$, $k=1,2,3$, are needed: $\hat{P}(A=1|W=\cdot)$, $\hat{P}(Y=1|S=s_1^\star,A=1,W=\cdot)$, $\hat{P}(S=s_1^\star|A=1,W=\cdot)$, and $\hat{P}(Y=0,S^c=s_1^\star|A=0,W=\cdot)$.
\end{itemize}
	\Comment{If estimated using loss-based learning, include observation weights $\Delta_i\bar{\pi}_i^{-1}$.}}
	\For{$k=1,2,3$}\Comment{Recall the definitions of $f_k$ and $a_k$ from Theorem~\ref{thm:eif}.}
		\State \parbox[t]{\dimexpr\linewidth-\algorithmicindent-\algorithmicindent}{\textbf{Fluctuation for Targeting Step:} Using observations $i=1,\ldots,n$, fit the intercept $\hat{\epsilon}_k$ using an intercept-only logistic regression with outcome $f_k(O_i)$, offset $\logit \hat{P}\left\{f_k(O)=1|a,w\right\}$, and weights $\frac{\Ind_{A_i=a_k}}{\hat{P}(A_i|W_i)}\Delta_i\bar{\pi}_i^{-1}$.}
		\State \parbox[t]{\dimexpr\linewidth-\algorithmicindent-\algorithmicindent}{\textbf{Targeting Step:} We now define a fluctuation of $\hat{P}_k^*$ of $\hat{P}$, targeted towards estimation of $\psi_k$. Define 
		\begin{align*}
		\hat{P}_k^*\left\{f_k(O)=1|a_k,w\right\}&\equiv \expit\left\{\logit \hat{P}\left\{f_k(O)=1|a_k,w\right\} + \hat{\epsilon}_k\right\}.
		\end{align*}
		The marginal distribution of $W\sim \hat{P}$ does not need to be fluctuated.
			}
		\State \parbox[t]{\dimexpr\linewidth-\algorithmicindent-\algorithmicindent}{\textbf{Plug-In Estimator:} The estimator of $\psi_k$ is $\hat{\psi_k}\equiv \Psi_k(\hat{P}_k^*)$, and the estimate of the corresponding component of the (full data) influence function is $\hat{D}_k\equiv D_k^{\hat{P}_k^*}$.\\
		\Comment{Here, by ``full data'' influence function, we mean that $\hat{D}_k$ is the influence function that we would have had if the phase two covariates had been measured on all subjects. When subsequently developing confidence intervals, $\hat{D}_k(O_i)$ will be weighted by $\Delta_i\bar{\pi}_i^{-1}$.}
		}
	\EndFor
	\State \parbox[t]{\dimexpr\linewidth-\algorithmicindent}{\textbf{return} Estimate $\boldsymbol{\hat{\psi}}\equiv (\hat{\psi}_k : k=1,2,3)$ and estimated full data influence function $\mb{\hat{D}}\equiv (\hat{D}_k : k=1,2,3)$.}
\EndFunction
\end{algorithmic}
\end{algorithm}

\section{Additional Simulation Results} \label{app:sim}
\begin{table}[ht]
\centering
\begin{tabular}{lrrrrrrrrrrr}
  \hline
$s_1^*$ & 0 & 0.1 & 0.2 & 0.3 & 0.4 & 0.5 & 0.6 & 0.7 & 0.8 & 0.9  \\ 
  \hline
  Bias, Truth & 0.03 & 0.04 & 0.04 & 0.03 & 0.03 & 0.03 & 0.03 & 0.05 & 0.06 & 0.08 \\ 
  Bias, Smoothed & 0.04 & 0.05 & 0.02 & -0.01 & -0.01 & 0.02 & 0.05 & 0.08 & 0.07 & 0.06 \\ 
  Standard Error & 0.28 & 0.27 & 0.26 & 0.26 & 0.27 & 0.29 & 0.31 & 0.34 & 0.39 & 0.49 \\ 
  Coverage, Truth & 0.97 & 0.97 & 0.98 & 0.98 & 0.97 & 0.97 & 0.97 & 0.97 & 0.97 & 0.96 \\ 
  Coverage, Smoothed & 0.97 & 0.97 & 0.97 & 0.96 & 0.96 & 0.97 & 0.97 & 0.98 & 0.97 & 0.96 \\ 
   \hline
\end{tabular}
\caption{Bias, standard error, and coverage probability of the log relative risk point estimator and confidence intervals for different $s_1^*$ values with bandwidth chosen to be $h=0.1$.}
\label{tab:h0.1}
\end{table}

\begin{table}[ht]
\centering
\begin{tabular}{lrrrrrrrrrrr}
  \hline
$s_1^*$ & 0 & 0.1 & 0.2 & 0.3 & 0.4 & 0.5 & 0.6 & 0.7 & 0.8 & 0.9  \\ 
  \hline
  Bias, Truth & -0.03 & -0.02 & 0.00 & 0.02 & 0.04 & 0.06 & 0.08 & 0.10 & 0.13 & 0.15 \\ 
  Bias, Smoothed & 0.03 & 0.02 & 0.02 & 0.02 & 0.01 & 0.01 & 0.00 & -0.00 & -0.00 & 0.00 \\ 
  Standard Error & 0.14 & 0.13 & 0.13 & 0.13 & 0.13 & 0.14 & 0.14 & 0.15 & 0.17 & 0.19 \\ 
  Coverage, Truth & 0.93 & 0.94 & 0.94 & 0.95 & 0.95 & 0.94 & 0.93 & 0.93 & 0.92 & 0.92 \\ 
  Coverage, Smoothed & 0.96 & 0.95 & 0.94 & 0.95 & 0.95 & 0.95 & 0.96 & 0.96 & 0.95 & 0.95 \\ 
   \hline
\end{tabular}
\caption{Bias, standard error, and coverage probability of the log relative risk point estimator and confidence intervals for different $s_1^*$ values with bandwidth chosen to be $h=0.3$.}
\label{tab:h0.3}
\end{table}

\begin{table}[ht]
\centering
\begin{tabular}{lrrrrrrrrrrr}
  \hline
$s_1^*$ & 0 & 0.1 & 0.2 & 0.3 & 0.4 & 0.5 & 0.6 & 0.7 & 0.8 & 0.9  \\ 
  \hline
  Bias, Truth & -0.07 & -0.04 & -0.01 & 0.02 & 0.05 & 0.08 & 0.12 & 0.15 & 0.18 & 0.22 \\ 
  Bias, Smoothed & 0.00 & 0.00 & 0.00 & 0.00 & 0.00 & 0.00 & 0.00 & 0.00 & -0.00 & -0.00 \\ 
  Standard Error & 0.11 & 0.11 & 0.11 & 0.11 & 0.11 & 0.11 & 0.12 & 0.13 & 0.14 & 0.15 \\ 
  Coverage, Truth & 0.90 & 0.93 & 0.95 & 0.95 & 0.93 & 0.91 & 0.84 & 0.79 & 0.73 & 0.71 \\ 
  Coverage, Smoothed & 0.95 & 0.96 & 0.95 & 0.95 & 0.95 & 0.95 & 0.95 & 0.95 & 0.95 & 0.96 \\ 
   \hline
\end{tabular}
\caption{Bias, standard error, and coverage probability of the log relative risk point estimator and confidence intervals for different $s_1^*$ values with bandwidth chosen to be $h=0.4$.}
\label{tab:h0.4}
\end{table}

\end{document}